\documentclass[11pt]{article}

\usepackage{amsmath,amsthm,latexsym,amsfonts,amstext,url,color,enumerate,setspace,epsfig}
\usepackage[]{graphicx}
\usepackage[numbers]{natbib}

\newtheorem{definition}{Definition}
\newtheorem{theorem}{Theorem}
\newtheorem{lemma}{Lemma}

\newenvironment{rlemma}[2][LEMMA]{%
#1 \ref{#2}. \begin{itshape}}{\end{itshape} 
}

\newenvironment{rtheorem}[2][THEOREM]{%
#1 \ref{#2}. \begin{itshape}}{\end{itshape} 
}

\newtheorem{assumption}{Assumption}

\newcommand{\xbi}{\overline{x}_i}
\newcommand{\zbi}{\overline{z}_i}

\newcommand{\vis}{v_i^*}

\newcommand{\ep}{\epsilon}

\newcommand{\g}{\gamma}

\newcommand{\la}{\leftarrow}

\newcommand{\argmax}{\mathrm{argmax}}

\newcommand{\hide}[1]{}
\hide{
     \newcommand{\qed}{\nobreak \ifvmode \relax \else
     \ifdim\lastskip<1.5em \hskip-\lastskip
     \hskip1.5em plus0em minus0.5em \fi \nobreak
     \vrule height0.75em width0.5em depth0.25em\fi}
}

\title{\textbf{Tatonnement in Ongoing Markets of Complementary Goods}}

\author{Yun Kuen Cheung \thanks{Courant Institute, New York University}
\and
Richard Cole \thanks{Courant Institute, New York University}
\and
Ashish Rastogi \thanks{Goldman Sachs \& Co., New York}
}

\begin{document}

\maketitle

\begin{abstract}\setlength{\parskip}{.1in}\setlength{\parindent}{0in}

\noindent This paper continues the study,
initiated by Cole and Fleischer in~\cite{CF2008},
of the behavior of a tatonnement price update rule in Ongoing Fisher Markets.
The prior work showed fast convergence toward an equilibrium when
the goods satisfied the weak gross substitutes
property and had bounded demand and income elasticities.

The current work shows that fast convergence also
occurs for the following types of markets:
\begin{itemize}
\item
All pairs of goods are complements to each other, and
\item
the demand and income elasticities are suitably bounded.
\end{itemize}
In particular, these conditions hold when all buyers in
the market are equipped with CES utilities,
where all the parameters $\rho$, one per buyer, satisfy $-1 < \rho \le 0$.

\smallskip

In addition,
we extend the above
result to markets in which a mixture of complements and substitutes
occur.
This includes characterizing a class of nested CES utilities for which
fast convergence holds.

An interesting technical contribution, which may be of independent interest,
is an amortized analysis for handling asynchronous events in settings in
which there are a mix of continuous changes and discrete events.
\end{abstract}

\section{Introduction}\label{sec:intro}

This paper continues the investigation, begun in~\cite{CF2008,ashish-thesis} 
of when a tatonnement-style price update in a dynamic
market setting could lead to fast convergent behavior.
This paper shows that there is a class of markets with complementary
goods which enjoy fast convergence; prior results applied only to
goods which are substitutes.

The impetus for this work comes from the following question: why might
well-functioning markets be able to stay at or near equilibrium
prices?
This raises two issues: what are
plausible price adjustment mechanisms and in what types of markets
are they effective?

This question was considered in~\cite{Walras1874}, which
suggested that prices adjust by tatonnement: upward if there is too
much demand and downward if too little. Since
then, the study of market equilibria, their existence, stability,
and their computation has received much attention in economics,
operations research, and most recently in computer science.
A fairly recent account of the
classic perspective in economics is given in~\cite{mckenzie}.
The recent activity in computer science
has led to a considerable number of polynomial time algorithms
for finding approximate and exact equilibria in a variety of markets
with divisible
goods; we cite a selection of these
works~\cite{CodenottiMV05,leontief,DevanurV04,DevanurPSV02,DevKan08,GargKapoor04,JainVaz07,Orlin10,vazirani10,VazYan10,Zhang2010}.
However, these algorithms do not seek to, and
do not appear to provide methods that might plausibly explain these
markets' behavior.

We argue here for the relevance of this question from a computer
science perspective. Much justification for looking at the market
problem in computer science stems from the following argument: If
economic models and statements about equilibrium and convergence are
to make sense as being realizable in economies, then they should be
concepts that are computationally tractable.
Our viewpoint is that
it is not enough to show that the problems are computationally
tractable; it is also necessary to show that they are tractable in a
model that might capture how a market works.
It seems implausible that markets with many interacting players
(buyers, sellers, traders) would perform overt global
computations, using global information.

Recently, a different perspective was put forward in~\cite{EchGolWie11},
which argues that
the fundamental question is whether computationally tractable instances of the model
can fit rational data sets.
But at this point, to the best of our knowledge, there are no results for
the market problem.

It has long been recognized that the tatonnement price adjustment
model is not fully realistic: for example, \cite{fisher72}
states:
``such a model of price adjustment $\cdots$ describes nobody's actual behavior.''
Nonetheless, there has been a continued interest in the
plausibility of tatonnement, and indeed its predictive accuracy
in a non-equilibrium trade setting has been shown in some experiments~\cite{hirota05}.

Plausibly, in many consumer markets buyers are myopic:
based on the current prices, goods are assessed on a take
it or leave it basis.
It seems natural that this would lead to out of equilibrium trade.
This is the type of setting which was studied in~\cite{CF2008} and
in which we will address our main question:
under what conditions can tatonnement style price updates lead
to convergence?

\subsection{The Market Problems}
\label{sec:problem}

{\bf The One-Time Fisher Market}\footnote{The market we describe here is
often referred to as the
Fisher market.  We use a different term because we
want to distinguish it from the Ongoing Fisher Market model described below.}
A market comprises a set of goods
$G= \{G_1, G_2, \cdots, G_n\}$, and two sets of agents,
buyers $B= \{B_1, B_2, \cdots, B_m\}$, and sellers $S$.
The sellers bring the goods $G$ to market and the buyers bring money with which to buy them.
The trade is
driven by a collection of prices $p_i$ for good $G_i$, $1\le i\le n$.
For simplicity, we assume that there is a distinct seller for each good;
further it suffices to have one seller per good.
The seller of $G_i$ brings a supply $w_i$ of this good to market.
Each seller seeks to sell its goods for money at the prices $p_i$.

Each buyer $B_\ell$ comes to market with money $b_\ell$;
$B_\ell$ has a utility function $u_\ell(x_{1\ell},\cdots,x_{n\ell})$
expressing its preferences: if $B_\ell$ prefers a basket with $x_{i\ell}$
units to the basket with $y_{i\ell}$ units, for $1\le i\le n$,
then $u_\ell(x_{1\ell},\cdots,x_{n\ell})> u_\ell(y_{1\ell},\cdots,y_{n\ell})$.
Each buyer $B_\ell$ intends to buy goods costing at most $b_\ell$ so
as to achieve a personal optimal combination (basket) of goods.

Prices ${p}=(p_1,p_2,\cdots,p_n)$ are said to provide an \emph{equilibrium}
if, in addition, the demand for each good is bounded by the supply:
$\sum_{\ell=1}^m x_{i\ell} \le  w_i$, and if $p_i>0$ then $\sum_{\ell=1}^m x_{i\ell} = w_i$.

The market problem is
to find equilibrium prices.\footnote{Equilibria exist under quite
mild conditions (see \cite{MWG95} \S 17.C, for example).}
The symbol$\ ^*$ marking a variable will be used to denote the value of the
variable at equilibrium.

\medskip

The Fisher market is a special case of the more general
\emph{Exchange} or \emph{Arrow-Debreu} market.

While we define the market in terms of a set of buyers $B$, all that matters for
our algorithms is the aggregate demand these buyers generate, so we will
tend to focus on properties of the aggregate demand rather than properties
of individual buyers' demands.

\smallskip \noindent {\bf Standard notation}~
$x_i=\sum_l x_{il}$
is the demand for good $i$, and $z_i=x_i-w_i$ is the excess demand
for good $i$ (which can be positive or negative).
$s_i = p_i x_i$ is the total spending on good $i$ by all buyers.
Note that while $w$ is part of the specification of the market, $x$
and $z$ are functions of the vector of prices as determined by individual buyers maximizing their
utility functions given their available money.
We will assume that $x_i$ is a function of the prices $p$, that is a collection of prices
induce unique demands for each good.

\smallskip

In order to have a realistic setting for a price adjustment algorithm,
it would appear that out-of-equilibrium trade must be allowed,
so as to generate the demand imbalances that then induce price adjustments.
But then there needs to be a way to handle
excess supply and demand. To this end, we
suppose that for each good there is a \emph{warehouse} which can
meet excess demand and store excess supply.   Each seller has a warehouse of
finite capacity to enable it to cope with fluctuations in demand. It
changes prices as needed to ensure its warehouse neither
overfills nor runs out of goods.

\smallskip\noindent
{\bf The Ongoing Fisher Market}~\cite{CF2008}~
The market consists of a set $G$ of $n$ goods and a set $B$ of $m$
buyers.
The seller of good $i$, called $S_i$, has a warehouse of capacity $\chi_i$.
As before, each buyer $B_{\ell}$ has a utility
function $u_\ell(x_{1\ell},\cdots,x_{n\ell})$ expressing its
preferences.
The market repeats over an unbounded number of time intervals called
\emph{days}. Each day, each seller $S_i$ receives
$w_i$ new units of good $i$, and each buyer $B_{\ell}$ is given $b_\ell$
money. Each day, $B_{\ell}$ selects a maximum utility basket
of goods $(x_{1\ell},\cdots,x_{n\ell})$ of cost at most $b_\ell$.
$S_i$, for each good $i$, provides the demanded goods $x_i = \sum_{\ell=1}^m
x_{i\ell}$. The resulting excess demand or surplus, $z_i = x_i - w_i$,
is taken from or added to the warehouse stock.

Given initial prices $p^\circ_i$, warehouse stocks $v^\circ_i$,
where $0<v^\circ_i<\chi_i$, $1\le i\le n$, and ideal warehouse stocks
$v^*_i$, $0<v^*_i<\chi_i$, the task is to repeatedly adjust prices so
as to converge to equilibrium prices with the warehouse stocks
converging to their ideal values; for simplicity, we suppose that
$v^*_i = \chi_i/2$.

We suppose that it is the sellers that are adjusting the prices of their goods.
In order to have progress, we require them to change prices at least once a day.
However, for the most part, we impose no upper bound on the frequency of price changes.
This entails measuring
demand on a finer scale than day units. Accordingly,
we assume that each buyer spends their money at a uniform
rate throughout the day. (Equivalently, this is saying that buyers
with collectively identical profiles occur throughout the day,
though really similar profiles suffice for our analysis.)~
If one supposes there is a limit to the granularity, this
imposes a limit on the frequency of price changes.

\smallskip\noindent {\bf Notation}~
We let $v_i$ denote the current
contents of warehouse $i$, and $v^e_i=v_i-v^*_i$ denote the
\emph{warehouse excess} (note that $v_i^e \in [-\chi_i/2, \chi_i/2]$).

\smallskip\noindent {\bf Market Properties}~ We recall from~\cite{CF2008}
that the goal of the ongoing market is to capture the
distributed nature of markets and the possibly limited knowledge
of individual price setters.
One important aspect is that the price updates for distinct goods
are allowed to occur independently and asynchronously.
We refer the reader to the prior work for a more extensive discussion.

\subsection{Our Contribution}

As it is not possible to devise a tatonnement-style price update for general markets
~\cite{papa-yann-10,SaariSimon78,Scarf60},
the goal, starting in~\cite{CF2008},
has been to devise plausible constraints that enable rapid convergence.
This entails devising (i) a reasonable model, (ii) a price update
algorithm, (iii) a measure of closeness to equilibrium, and then (iv) analyzing
the system to demonstrate fast convergence; (v) this also entails
identifying appropriate constraints on the market.
(i)--(iii) were addressed in~\cite{CF2008}, though there was an unsatisfactory element
to the price update rule when coming close to breaching a warehouse bound
(i.e.\ the warehouse becoming empty or full); this is fixed in the current work.

The constraints in~\cite{CF2008} were for markets of substitutes.
These constraints take the form of bounds on the elasticities of demand and wealth
(defined later).
Curiously, the best performance occurred at the boundary between substitutes and
complements (with the buyers having so-called Cobb-Douglas utilities).
Despite this, the result did not extend into the complements domain.
The present paper carries out such an extension,
handling markets in which a mixture of complements and substitutes
occur.
The markets for which we show convergence again have bounded elasticities
(the precise constraints are detailed later).
These markets include the following scenarios.
\begin{enumerate}
\item
\label{first-item}
All the goods are complements.
A particular instance of this setting occurs when each buyer has a
CES utility with
parameter $\rho$ satisfying $-1<\rho \le 0$ (defined later).
\item
\label{second-item}
(A generalization of~(\ref{first-item}))~
The goods are partitioned into groups. Each group comprises substitutes, while
the groups are complementary.
A particular instance of this setting occurs when each buyer has
a suitable 2-level nested
CES utility~\cite{Keller1976} (defined later).
\item
Each buyer has a suitable arbitrary depth nested CES utility.
\end{enumerate}
Overall, we believe this result significantly expands the class of markets
for which the rapid convergence of tatonnement is known to hold,
and thereby enhances the plausibility of tatonnement as a usable
price adjustment mechanism.

There are relatively few positive results for markets of complementary goods,
and to the best of our knowledge none for tatonnement algorithms.
\cite{CMPV2005} gave a polynomial time algorithm based on
convex programming to compute
equilibrium prices for an Exchange Market in which every agent has
a CES utitility with $\rho$ in the range $-1\le \rho \le 0$.
\cite{CodenottiV-ICALP04} gave a polynomial time algorithm for Fisher markets
with Leontief utilities, which was generalized
in~\cite{chen-teng08}, which considered hybrid linear-Leontief utility functions,
in which goods are grouped, and within a group the utilities are Leontief, and
the group utilities are combined linearly.

The economics literature has many results concerning tatonnement, but the positive
results largely concerned markets in which utilities satisfied weak gross substitutes,
i.e.\ the goods were substitutes.

Finally we discuss the amortized analysis technique we introduce to handle asynchronous events.
We use a potential function $\phi$ which satisfies two properties:
\begin{itemize}
\item
$\frac{d\phi}{dt} \leq -\Theta(\kappa)\phi$ for a suitable parameter $\kappa >0$ whenever there is
no event.
\item
$\phi$ is non-increasing when an event occurs (a price update in our application).
\end{itemize}
One can then conclude that $\phi(t) \le e^{-\Theta(\kappa) t}\phi(0)$, and so $\phi$ decreases by
at least a $1-\Theta(\kappa)$ factor daily (for $\kappa =O(1)$).

It is not clear how to craft a more standard amortized analysis, in which $\phi$ changes only
when an event occurs.
The difficulty we face is that we model the warehouse imbalances as changing continuously,
and it is not clear how to integrate the resulting cost with the gains
from the price update events if $\phi$ changes only discretely.

\subsection{Roadmap}

In Section~\ref{sec:prelim} we state the price update rules
and review the definitions of elasticity.
We state our main results in Section~\ref{sec:results},
Then in Section~\ref{sec:anal-outline} we provide an outline of the analysis for markets where all the
goods are complementary, illustrating this with the scenario in which
every buyer has a CES utility function.
Finally, in Section~\ref{sec:mixed},
we analyze the mixed complements and substitutes scenario, illustrating
it first with markets in which the buyers all have 2-level nested CES-like utilities,
and then expanding the result to arbitrary levels of nesting.
Some of the proofs are deferred to the appendix.

\section{Preliminaries}\label{sec:prelim}

We review the price update rule and the definitions of elasticities.
The basic price update rule for the one-time market, proposed in~\cite{CF2008},
is given by
\begin{equation}
\label{eq:tat-rule}
p_i \la p_i \left( 1 + \lambda\cdot\min\left\{1,\frac{x_i-w_i}{w_i}\right\} \right),
\end{equation}
where $0 < \lambda \leq 1$ is a suitable parameter whose value depends on the market
elasticities.

In the ongoing market, the excess demand is computed as the excess demand since the
previous price change at time $\tau_i$. Thus in Equation~\ref{eq:tat-rule}, $x_i$ is
replaced by the average demand since time $\tau_i$,
$\bar{x}_i[\tau_i,t] = \frac{1}{t-\tau_i} \int_{\tau_i}^t x_i(t)\,dt$,
where $t$ is the current time.
Recall we assumed that each seller adjusts the price of its good at least once each day,
so $t-\tau_i\leq 1$.

In addition, one needs to take account of the warehouse excess, with prices dropping
if there is too much stock in the warehouse, and increasing if too little.
To this end, we define the
\emph{target demand} $\widetilde{w}_i$ to be
$$\widetilde{w}_i := w_i + \kappa (v_i - v_i^*),$$
where $\kappa>0 $ is chosen to ensure that
$|\kappa (v_i - v_i^*)| \leq \delta w_i$,
for a suitable $\delta>0$.

Now, as in~\cite{CF2008}, we define the \emph{target excess demand} to be
$$\bar{z}_i := \bar{x}_i [\tau_i,t] - \widetilde{w}_i =
  \bar{x}_i [\tau_i,t] -w_i - \kappa(v_i-v_i^*).$$

As it takes time for the warehouse stock to adjust as a result of a price
change, it turns out that the price change needs to be proportional to the
time since the last price update (this is where the price update rule differs
from~\cite{CF2008}).
This yields the price update rule
\begin{equation}\label{eq:tat-rule-warehouse}
p_i \la p_i \left( 1 + \lambda\cdot\min\left\{1,\frac{\bar{z}_i}{w_i}\right\}\Delta t \right),
\end{equation}
where $\Delta t$ is the time since the last price update.

Implicitly, the price update rule is using a linear approximation
to the relationship between $p_i$ and $x_i$.
The analysis would be cleaner if the linear update were to $\log p_i$;
however, this seems a less natural update rule, and having a natural
rule is a key concern when seeking to argue tatonnement could be a real process.

Next, we review the definitions of income and price elasticity.

\begin{definition}
\label{def:income-elas}
The \emph{income elasticity} of the demand for good $i$
by a buyer with income (money) $b$ is given by $\frac{dx_i}{db}\left/ \frac{x_i}{b}\right.$.
We let $\g$ denote the least upper bound on the income elasticities over all buyers and goods.
\end{definition}
If all buyers are spending their budgets in full, then $\g \ge 1$.

\begin{definition}
\label{def:price-elas}
The \emph{price elasticity} of the demand for good $i$ is given by $-\frac{dx_i}{dp_i} \left/ \frac{x_i}{p_i}\right.$.
We let $\alpha$ denote the greatest lower bound on the price elasticities over all goods.
\end{definition}

In a market of complementary goods $0 \le \alpha \le 1$;
in the markets we consider, $\alpha > \gamma/2 \ge 1/2$. 

For the markets with mixed complements and substitutes
we need a generalized version of elasticity, which we call the
\emph{Adverse Market Elasticity}.
These are the extreme changes in demand that occur to one good, WLOG $G_1$,
when its price changes, and other prices also change but by no larger a
fraction than $p_1$.
For suppose that $p_1$ were reduced with the goal of increasing $x_1$.
But suppose that at the same time other prices may change by the same fractional
amount (either up or down).
How much can this undo the desired increase in $x_1$?
The answer is that in general it can more than undo it.
However, our proof approach 
depends on $x_1$
increasing in this scenario, which is why we introduce this notion of elasticity
and consider those markets in which it is sufficiently bounded from below.

\begin{definition}
\label{def:den-elas}
Define $\bar{P}$ to be the following set of prices:
$\left\{((1+\delta)p_1,q_2, \cdots, q_n)~|~ \mbox{for $i \ge 2$, }q_i \in [\frac{p_i}{1+\delta},(1+\delta)p_i]\right\}$.
The (low) \emph{Adverse Market Elasticity} for $G_1$ is defined to be
$$-\max_{\bar{p} \in \bar{P}} \lim_{\delta \rightarrow 0} \frac {x_1(\bar{p}) - x_1(p) } {\delta x_1}$$
We let $\beta$ be a lower bound on the Adverse Market Elasticity
over all goods and prices.
\end{definition}
It is not hard to see that for the case that all the goods are complements, $\beta \ge 2\alpha-\gamma$.

\section{Our Results}\label{sec:results}

\subsection{Markets with Complementary Goods}

The analysis of these markets introduces the techniques needed for the more general setting.

Our bounds will depend on the market parameters $\alpha$ and $\gamma$.
It is convenient to set $\beta = 2 \alpha - \gamma$; note that by assumption, $0 < \beta \le 1$.
In addition, our bounds will depend on the initial imbalance in the prices. To specify this we
define the notion of $f$-bounded prices.

\begin{definition}\label{def:f-bounded}
$$f_i(p) = \max\left( \frac{p_i}{p_i^*},\frac{p_i^*}{p_i} \right),$$
and
$$f(p) = \max_{1\leq i\leq n} f_i(p).$$
The prices are \emph{$f$-bounded} if $f(p) \leq f$.
\end{definition}
Clearly, $f(p)\geq 1$ and $f(p)=1$ if and only if $p=p^*$.
When there is no ambiguity, we use $f$ as a shorthand for $f(p)$.
We let $f_{\mbox{\tiny I}}$ denote the maximum value $f$ takes on
during the first day, which will also bound $f$ thereafter, as we will see.

It will turn out that we can assume $\chi_i/w_i$ are equal for all $i$;
we denote this ratio by $r$.

Our results will require $\lambda$ and $\kappa$ to obey the following conditions.

\begin{equation}\label{eqn:kappa-cond1}
\kappa \leq \frac{2}{r}\cdot\min\left\{ \frac{\beta}{4\gamma + \beta}, \frac{1}{2(8 + 4\gamma/\beta)}\ln\frac{1}{2(1-\alpha)}\right\}
\end{equation}
\begin{equation}\label{eqn:lambda-cond1}
\frac{24}{r} \leq \lambda \leq \min\left\{\frac{3}{7},\frac{3}{7}\ln\frac{1}{2(1-\alpha)},\sqrt{\frac{\kappa r}{32}}\right\}.
\end{equation}

We then show the following bound on the convergence rate.

\begin{theorem}
\label{th:comp-overall}
Suppose that $\beta>0$, the prices are $f$-bounded throughout the first day,
and in addition that Equations~\eqref{eqn:kappa-cond1}--\eqref{eqn:lambda-cond1} hold.
Let $M= \sum_j b_j$ be the daily supply of money to all the buyers.
Then the prices become $(1+\eta)$-bounded after
$O\left(\frac{1}{\lambda}\ln f + \frac{1}{\lambda \beta}\ln\frac{1}{\delta} + \frac{1}{\kappa} \log \frac{M}{\eta \min_i w_i p_i^*}\right)$ days.
\end{theorem}

We also bound the needed warehouse sizes.
We say that warehouse $i$ is \emph{safe} if $v_i \in [\frac 14 \chi_i, \frac 34 \chi_i]$.
We define $d(f) = \max_i x_i/w_i$ when the prices are $f$-bounded.
In a market of complementary goods, $d(f) \le f^{\gamma}$.

As we will see, the analysis of Theorem~\ref{th:comp-overall} proceeds in two phases.
We need to specify some parameters relative to Phase 1.
We define $v(f_{\mbox{\tiny I}})$ to be the total net change to $v_i$ during Phase 1.
As we will see,
$v(f_{\mbox{\tiny I}}) = O(\frac{w_i}{\lambda} d(f_{\mbox{\tiny I}}) + \frac{w_i}{\lambda\beta} d(2) \log \frac{\beta}{\delta})$.
We define $D(f)$ to be the duration of Phase 1 in days.
As we will see,
$ D(f) = O\left(\frac{1}{\lambda}\ln f + \frac{1}{\lambda \beta}\ln\frac{1}{\delta}\right)$.
We will show:

\begin{theorem}
\label{lem:good-wrhs}
Suppose that the ratios $\chi_i/w_i$ are all equal.
Suppose that the prices are always $f$-bounded.
Also suppose that each price is updated at least once a day.
Suppose further that initially the warehouses are
all safe.
Finally, suppose that Equations~\eqref{eqn:kappa-cond1}--\eqref{eqn:lambda-cond1} hold.
Then the warehouse stocks never either overflow
or run out of stock;
furthermore,
after
$D(f) + \frac {32}{\beta} + \frac{2}{\kappa}$
days the warehouses will be safe thereafter.
\end{theorem}

\subsubsection{Example Scenario: All buyers have CES Utilities with $-1 < \rho \le 0$}

We begin by reviewing the definition of CES utilities.
We focus on a single buyer $B_\ell$, and to simplify notation, we let
$(y_1,y_2, \cdots, y_n) = (x_{1\ell}, x_{2\ell} \cdots, x_{n\ell})$.
A CES utility has the form
$$u(y_1,y_2,\cdots,y_n) = \left(\sum_{i=1}^n a_i y_i^{\rho_\ell}\right)^{1/\rho_\ell}.$$
It is well known that when $\rho_\ell\ge 0$, all goods are substitutes,
and when $\rho_\ell\le 0$, all goods are complements. We will focus on the case $-1<\rho_\ell \le 0$.
It is also well known that with a budget constraint of $b$, the resulting demands are given by
$y_i = p_i^r b a_i^{-r}\left(\sum_{j=1}^n \frac{p_j^{r+1}}{a_j^r}\right)^{-1}$, where $r =1/(\rho_\ell-1)$.
A further calculation yields that
$\frac{\partial y_i / \partial p_i}{y_i / p_i} \le  r = -1/(1-\rho_\ell)$.
Let $\rho = \min_\ell \rho_\ell$.
Then $\frac{\partial x_i / \partial p_i}{x_i / p_i} \le -1/(1-\rho)$.
In addition, it is easy to show that for CES utilities, $\gamma = 1$.
Consequently, when $\rho > -1$, $\beta > 0$; it follows that the bounds
from Theorems~\ref{th:comp-overall} and~\ref{lem:good-wrhs} apply.

\subsection{Markets with Mixtures of Substitutes and Complements}\label{sect:result-mix}

To understand the constraints needed in this setting, we need to know that the
analysis for markets of complements, which we adapt to the current setting, proceeds
in two phases. Recall that the prices are $f$-bounded.
In Phase 1,  $(f-1)$ reduces multiplicatively each day.
Phase 1 ends when further such reductions can no longer be guaranteed.
In Phase 2, an amortized analysis shows that the misspending,
roughly $\sum_i |z_i| p_i + |\widetilde{w}_i - w_i| p_i$, decreases
multiplicatively each day.

Also, now that substitutes are present, we will need an
upper bound on the price elasticity (see Definition \ref{def:price-elas}),
as in~\cite{CF2008}.
We let $E\geq 1$ denote this upper bound.
For convergence we will need that
$\lambda = O(1/E)$.

Denote the spending on all goods which are substitutes of $G_1$ by $S_s$
and the spending on all goods which are complements of $G_1$ by $S_c$.
We need to introduce a further constraint.
The reason is that the amortized analysis depends on showing the misspending decreases.
However, the current constraints do not rule out the possibility that
when, say $p_1$ is increased, the spending decrease on $G_1$'s complements,
$|\Delta S_c|$, and the spending increase on $G_1$'s substitutes, $|\Delta S_s|$,
are both larger than the reduction in misspending on $G_1$.
It seems quite unnatural for this to occur.
We rule it out with the following assumption.

\begin{assumption}
\label{ass:sp-trans}
Suppose that $p_i$ changes by $\Delta p_i$.
Then there is a  constant $\alpha' < \frac 12$, such that
$|\Delta S_c| \leq \alpha' x_i |\Delta p_i|$.
\end{assumption}

We require that $\beta$, as defined in Definition \ref{def:den-elas}, satisfy $\beta > 0$.
Our results will require $\lambda$ and $\kappa$ to obey the following conditions.

\begin{equation}\label{eqn:kappa-cond2}
\kappa \leq \frac{2}{r}\cdot\min\left\{ \frac{\beta}{\beta + 4(2E-\beta)},\frac{(1-2\alpha ')\beta}{8\alpha '(2E-\beta) + 4\beta}\right\}
\end{equation}

\begin{equation}\label{eqn:lambda-cond2}
\frac{24}{r}\leq \lambda \leq \min\left\{\frac{1}{8E + 4\alpha ' - 6},\sqrt{\frac{\kappa r}{32}}\right\}
\end{equation}

\begin{theorem}
\label{th:multi-overall}
Suppose that $\beta>0$, the prices are $f$-bounded throughout the first day,
and Equations~\eqref{eqn:kappa-cond2}--\eqref{eqn:lambda-cond2} hold.
Let $M= \sum_j b_j$ be the daily supply of money to all the buyers.
Then the prices become
$(1+\eta)$-bounded after
$$O\left(\frac{1}{\lambda}\ln f
+ \frac{1}{\lambda \beta}\ln\frac{1}{\delta} + \frac{1}{\kappa} \log \frac{M}{\eta \min_i w_i p_i^*}\right)$$
days.
\end{theorem}

Theorem~\ref{lem:good-wrhs} with Equations~\eqref{eqn:kappa-cond2}--\eqref{eqn:lambda-cond2}
replacing Equations~\eqref{eqn:kappa-cond1}--\eqref{eqn:lambda-cond1} also continues to apply.
Here $d(f) \le f^{2E - \beta}$.

\subsubsection{Example Scenario: 2-Level Nested CES Type Utilities}\label{subsect:2-level-nces}

Keller \cite{Keller1976} proposed and analyzed \emph{nested CES-type utility functions},
which we use to provide an example of utility functions
yielding a mixture of substitutes and complements.
Goods are partitioned into different groups.
Two goods in the same group are substitutes,
while two goods in different groups are complements.

Again, we focus on the demands $(y_{1\ell}, y_{2\ell}, \cdots, y_{n\ell})$ of a single buyer $B_\ell$.
For each group $\overline{G}$, we define
$$u_{\overline{G},\ell} := \left(\sum_{i\in \overline{G}} a_{i\ell} ~ y_{i\ell}^{\rho_{\overline{G},\ell}}\right)^{1/\rho_{\overline{G},\ell}},$$ which is called a \emph{utility component}; $0 < \rho_{\overline{G},\ell} < 1$.
The overall utility function is given by
$$u_\ell := \left(\sum_{{\overline{G}}} a_{\overline{G},\ell} ~ u_{\overline{G},\ell}^{\rho_\ell} \right)^{1/{\rho_{\ell}}},$$
where $-1 < \rho_\ell<0$.
The bounds on $\rho_{\overline{G},\ell}$ and $\rho_{\ell}$ are needed to allow us to show convergence; Keller allowed arbitrary
values (no larger than 1).

We will show that $E=\max_{\overline{G},\ell} \frac{1}{1-\rho_{\overline{G},\ell}}$ and $\beta = \min_{\ell}\frac{2}{1-\rho_{\ell}}-1$.
The bounds from Theorems~\ref{th:multi-overall} and~\ref{lem:good-wrhs} will apply.

\subsubsection{Example Scenario: \boldmath{$N$}-Level Nested CES Type Utilities}\label{subsect:N-level-nces}

This result extends to arbitrary levels of nesting. A Nested CES Type Utility is best visualized as a utility tree.
A leaf represents a good, and an internal node represents a \emph{utility component}.
There is a value of $\rho$ associated with each internal node.
Each utility component is of the form $\left(\sum_{k=1}^m a_k u_k^\rho\right)^{1/\rho}$,
in which $u_k$ may be the quantity of one good $x_k$ or another utility component.

We focus on one particular good $i$.
Let $A_1,A_2,\cdots,A_N$ be the internal nodes along the path from good $i$ to the root of the utility tree,
and let $\rho_1,\rho_2,\cdots,\rho_N$ be the associated $\rho$ values. Let $\sigma_k = \frac{1}{1-\rho_k}$ for $1\leq k\leq N$.
Define $\beta_i = \sigma_1 - |\sigma_N-1| - \sum_{q=1}^{N-1} |\sigma_q - \sigma_{q+1}|$,
$E_i = \max\left\{1,\max_{1\leq k \leq N} \sigma_k\right\}$ and
$\alpha_i ' = (1-\lambda)^{-E} \left(|\sigma_N-1| + \sum_{q=1}^{N-1} |\sigma_q - \sigma_{q+1}|\right)$.
We set $\beta = \min_i \beta_i$, $E = \max_i E_i$ and $\alpha ' = \max_i \alpha_i '$.
Again, the bounds from Theorems~\ref{th:multi-overall} and~\ref{lem:good-wrhs} apply.

\subsection{Comparison to Prior Work}

\cite{CF2008} introduced the notion of ongoing markets and analyzed  a
class of ongoing Fisher markets satisfying WGS.
The current work extends this analysis to classes of ongoing Fisher markets
with respectively, only complementary goods, and with a mixture of complements and substitutes.
The present work also handles the warehouses in the ongoing model more realistically.

This entails a considerably changed analysis and some modest changes to the price update rule.
As in~\cite{CF2008}, the analysis proceeds in two phases.
The new analysis for Phase 1, broadly speaking, is similar to that in~\cite{CF2008},
though a new understanding was needed to extend it to the new markets.
The analysis for Phase 2 is completely new.
The analysis of the bounds on the warehouse sizes is also new.

A preliminary unrefereed report on these techniques, applied to
markets of substitutes, was given in~\cite{CFR10}; this manuscript
included other results too (on managing with approximate values of the demands,
and on extending the results to markets of indivisible goods).
The current paper subsumes the analysis techniques in~\cite{CFR10}.\footnote{Note
for the reviewers: This is the one refereed venue where these techniques are being
submitted for publication. We make this point because with an earlier submission of
this work, one referee appeared to consider~\cite{CFR10} to be prior work.}

\section{The Analysis for Complementary Goods}\label{sec:anal-outline}

The largest challenge in the analysis is to handle the effect of warehouses.
In~\cite{CF2008}, the price updates increased in frequency as the warehouse
limits (completely full or empty) were approached,
which ensured these limits were not breached.
It was still a non-trivial matter to demonstrate convergence.
In the present paper, the only constraint is that each price is updated
at least once every full day.
This seems more natural, but entails a different and new analysis.

The analysis partitions into two phases, the first one handling the
situation when at least some of the prices are far from equilibrium,
and for these prices, the warehouse excesses have a modest impact on
the updates.
This portion of the analysis is somewhat similar to the corresponding analysis
in~\cite{CF2008},
except that we manage to extend it to markets including complementary
goods.
In the second phase, the warehouse excesses can have a significant effect.
For this phase, we use a new and amortized analysis.
The imbalance being measured and reduced during Phase 2 is the \emph{misspending}
(roughly speaking, $\sum_i [p_i|x_i-w_i| + p_i |\widetilde{w}_i - w_i|]$).
It is only when prices are reasonably close to their equilibrium values
that we can show the misspending decreases, which is why two phases are needed.
Interestingly, in a market of substitutes, regardless of the prices,
the misspending is always decreasing, so here one could carry out the whole
analysis within Phase 2.

\paragraph{Phase 1}
In Phase 1, we show that each day
$(f-1)$ shrinks by a factor of at least $1-\Theta(\lambda\beta)$.

For simplicity, we begin by considering the one-time market.

Suppose that currently the prices are exactly $f$-bounded, and that there is
a good, WLOG good $G_1$, with price $p_1= p_1^*/f$.
We will choose the market properties to ensure that ${x}_1\ge f^{\beta}w_1$,
regardless of the prices of the other goods, so long as they are $f$-bounded.
This ensures that the price update for $p_1$ will be an increase, by a factor of
at least $1 + \lambda \min\{1, (f^{\beta}-1)\}) \doteq 1+ \mu$.

To demonstrate the lower bound on ${x}_1$, we identify the following scenario as the
one minimizing $x_1$: all the complements $G_i$ of $G_1$ have prices $fp^*_i$.

A symmetric observation applies when $p_1 = fp^*_1$, and then the price
decreases by a factor of at least
$1 - \lambda (1 - f^{-\beta}  ) \doteq 1 - \nu$.

We can show that the same market properties imply that after a day of price updates
every price will lie within the bounds
$[p^*(1+\mu)/f, fp^*(1-\nu)]$,
thereby ensuring a daily reduction of the term
$(f-1)$ by a factor of at least $1 - \Theta(\lambda\beta)$.

We use a similar argument for the ongoing market.
First, we observe that if the price updates
occurred simultaneously exactly once a day, then
exactly the same bounds would apply to $\bar{x}_i$, so the rate
of progress would be the same, aside the contribution of the warehouse
excess to the price update.
So long as this contribution
is small compared to $(f^{\beta} -1) w_1$ or to $(1-f^{-\beta} ) w_1$, say at most
half this value, then the price changes would still be by a factor of at least
$1 + \frac 12 \lambda \min\{1, (f^{\beta}-1)\})$ and
$1 - \frac 12 \lambda (1 - f^{-\beta} )$,
respectively.

To take account of the possible variability in price update frequency,
we demonstrate progress as follows:
we can show that if the prices have been $f$-bounded for a full day,
then after two more days have elapsed, the prices will have been $f'$-bounded
for a full day, for $(f'-1) = (1 -\Theta(\lambda\beta))(f-1)$.
The reason we look at the $f$-bound over the span of a day is that the price
updates are based on the average excess demand over a period of up to
one day.
A second issue we need to handle is that the price updates may have
a variable frequency; the only guarantee is that each price is
updated within one full day of its previous update.
The net effect is that it takes one day to guarantee that $f$ shrinks
and hence two days for the shrinkage to have lasted at least one full day.

It follows that Phase 1 lasts
$O(\frac{1}{\lambda\beta}\log[(f_{\mbox{\tiny I}}-1)/(f_{\mbox{\tiny II}}-1)])$ days,
where $f_{\mbox{\tiny I}}$ is the initial value of $f$,
and $f_{\mbox{\tiny II}}$ is its value at the start
of Phase 2.

We want the following conditions to hold in Phase 2:
$\bar{x}_i, x_i \le (2- \delta) w_i$ and $p_i \le 2 p^*_i$.
As we will see, choosing $f_{\mbox{\tiny II}} = \min\{(1 - 2\delta)^{-1/\beta}, (2-\delta)^{1/\gamma}\}$ suffices.
As it turns out,
the calculations simplify if we also enforce that $(1 - 2\delta)^{-1/\beta} \le  (2-\delta)^{1/\gamma}$.
If $2\delta/\beta \le \frac 12$,
then $1 + 2 \delta/\beta \le f_{\mbox{\tiny II}} \le 1 + 4 \delta/\beta \le 2$.

As already argued, the behavior of the ongoing and one-time markets are within
a constant factor of each other in Phase 1 (the ongoing market progresses
in cycles of two days rather than one day, and reduces $(f-1)$ by a factor
in which $\lambda$ is replaced by $\lambda/2$).
So in this phase we analyze just the one-time market.

First, we state several inequalities we use. They can be proved by simple arithmetic/calculus.

\begin{lemma}\label{lem:cal-result}
\begin{enumerate}[(a)]
\item If $0\leq\lambda\leq 1$, then $\frac{1}{1+\lambda} \leq 1-\frac{\lambda}{2}$.
\item If $0\leq\lambda\leq 1$ and $1\leq x\leq 2$, then $1-\lambda(1-1/x) \leq x^{-\lambda/(2\ln 2)}$.
\item If $0\leq\lambda\leq 1$ and $1\leq x\leq 2$, then $\frac{1}{1+\lambda(x-1)} \leq x^{-\lambda}$.
\end{enumerate}
\end{lemma}

Using the definitions of $\gamma$ and $\alpha$ in Definitions~\ref{def:income-elas} and \ref{def:price-elas},
one can easily obtain the following lemma.

\begin{lemma}
\begin{enumerate}[(a)]\label{lem:complementary-demand-bound-prelim}
\item If the prices of all goods are raised from $p_i$ to $p_i' = q p_i$, where $q>1$, then $x_i' \geq x_i / q^\gamma$.
\item If the prices of all goods are reduced from $p_i$ to $p_i' = q p_i$, where $q<1$, then $x_i' \leq x_i / q^\gamma$.
\item If $p_i$ is raised to $p_i' = t p_i$, where $t>1$, and all other prices are fixed, then $x_i' \leq x_i / t^\alpha$.
\item If $p_i$ is reduced to $p_i' = t p_i$, where $t<1$, and all other prices are fixed, then $x_i' \geq x_i / t^\alpha$.
\end{enumerate}
\end{lemma}

\begin{lemma}\label{lem:f-bounded-demand-bound-1}
When the market is $f$-bounded,
\begin{enumerate}
\item if $p_i = r p_i^* / f$ where $1\leq r\leq f^2$, then $x_i\geq w_i f^{2\alpha-\gamma} r^{-\alpha}$;
\item if $p_i = f p_i^* / q$ where $1\leq q\leq f^2$, then $x_i \leq w_i f^{\gamma-2\alpha} q^\alpha$.
\end{enumerate}
\end{lemma}
\begin{proof}
We prove the first part; the second part is symmetric. By the definition of complements, $x_i$ is smallest when $p_j = f p_j^*$ for all $j\neq i$. Consider the situation in which $p_k = f p_k^*$ for all goods $k$. By Lemma \ref{lem:complementary-demand-bound-prelim}(a), $x_i \geq \frac{w_i}{f^\gamma}$. Now reduce $p_i = f p_i^*$ to $p_i = r p_i^* / f$. By Lemma \ref{lem:complementary-demand-bound-prelim}(d), $x_i \geq \frac{w_i}{f^\gamma (r/f^2)^\alpha} = w_i f^{2\alpha - \gamma} r^{-\alpha}$.
\end{proof}

\begin{lemma}\label{lem:conv-one-time-market-1}
Suppose that $\beta = 2\alpha - \gamma > 0$.
Further, suppose that the prices are updated independently using price update rule \eqref{eq:tat-rule}, where $0<\lambda\leq 1$.
Let $p$ denote the current price vector and let $p'$ denote the price vector after one day.
\begin{enumerate}
\item If $f(p)^\beta \geq 2$, then $f(p') \leq \left(1-\frac{\lambda}{2}\right) f(p)$.
\item If $f(p)^\beta \leq 2$, then $f(p') \leq f(p)^{1-\lambda\beta/(2\ln 2)}$.
\end{enumerate}
\end{lemma}

\begin{proof}
Suppose that $p_i = r \frac{p_i^*}{f(p)}$, where $1\leq r\leq f(p)^2$. By Lemma \ref{lem:f-bounded-demand-bound-1}, $x_i \geq w_i f(p)^\beta r^{-\alpha}$ and hence $\frac{x_i-w_i}{w_i} \geq f(p)^\beta r^{-\alpha} - 1$.
When $p_i$ is updated using price update rule \eqref{eq:tat-rule}, the new price $p_i'$ satisfies
$$p_i' \geq r \frac{p_i^*}{f(p)} \left[1 + \lambda\cdot\min\left\{1,f(p)^\beta r^{-\alpha} - 1\right\}\right].$$
Let $h_1(r) := r \left[1 + \lambda\cdot\min\left\{1,(f(p)^\beta r^{-\alpha} - 1)\right\}\right]$.
Then
$$
\frac{d}{dr} h_1(r) \geq 1 - \lambda + (1-\alpha) \lambda f(p)^\beta r^{-\alpha} \geq 0.
$$
Thus
$$
p_i' \geq \frac{p_i^*}{f(p)} \left[1 + \lambda\cdot\min\left\{1,f(p)^\beta - 1\right\}\right].
$$

Similarly, suppose that $p_j = \frac{1}{q} f(p) p_j^*$, where $1\leq q\leq f(p)^2$. By Lemma \ref{lem:f-bounded-demand-bound-1}, $x_j \leq w_j f(p)^{-\beta} q^\alpha$ and hence $\frac{x_j-w_j}{w_j}\leq f(p)^{-\beta} q^\alpha-1$. When $p_j$ is updated using price update rule \eqref{eq:tat-rule}, the new price $p_j'$ satisfies
$$p_j' \leq \frac{1}{q} f(p) p_j^* \left[ 1 + \lambda\cdot\min\left\{ 1, f(p)^{-\beta} q^\alpha-1\right\} \right].$$
Let $h_2(q) := \frac{1}{q}\left[ 1 + \lambda\cdot\min\left\{ 1, f(p)^{-\beta} q^\alpha-1\right\} \right]$.
Then
$$
\frac{d}{dq} h_2(q) \leq \frac{1}{q^2}\left(\lambda - 1 - (1-\alpha) \lambda f(p)^{-\beta} q^\alpha\right)\leq 0.
$$
Thus
\begin{eqnarray*}
p_j' &\leq & f(p) p_j^* \left[ 1 + \lambda\cdot\min\left\{ 1, f(p)^{-\beta} -1\right\} \right] \\
&=& f(p) p_j^* \left[ 1 + \lambda (f(p)^{-\beta} -1) \right].
\end{eqnarray*}

Hence, after one day, a period in which each good updates its price at least once, we can guarantee that $f(p')$ is at most
$$f(p) \cdot \max\left\{ 1 - \lambda (1-f(p)^{-\beta}) , \frac{1}{1 + \lambda \cdot \min\left\{1,(f(p)^\beta - 1)\right\}} \right\},$$
which, by Lemma \ref{lem:cal-result}(a), is at most $\left(1-\frac{\lambda}{2}\right) f(p)$ when $f(p)^\beta \geq 2$.

When $f(p)^\beta\leq 2$, by Lemma \ref{lem:cal-result}(b), $1-\lambda(1-f(p)^{-\beta}) \leq f(p)^{-\lambda\beta/(2\ln 2)}$. By Lemma \ref{lem:cal-result}(c), $\frac{1}{1+\lambda(f(p)^\beta-1)} \leq f(p)^{-\lambda\beta}$.
So $f(p') \leq f(p)^{1-\lambda\beta/(2\ln 2)}$.
\end{proof}

\begin{theorem}
\label{th:phase-1-comps}
Suppose that $\beta > 0$,
$\lambda \le 1$, and that the prices are initially
$f$-bounded.
When $\delta/\beta \leq 1$, Phase 1 will complete within
$O\left(\frac{1}{\lambda}\ln f + \frac{1}{\lambda \beta}\ln\frac{1}{\delta}\right)$ days.
\end{theorem}

\begin{proof} Suppose that initially $f>2^{1/\beta}$ and $1 + 2\delta/\beta<2^{1/\beta}$.
By Lemma \ref{lem:conv-one-time-market-1}, after $n_1$ days, where $n_1$ satisfies the inequality
$f\left(1-\frac{\lambda}{2}\right)^{n_1} \leq 2^{1/\beta}$, the market is $2^{1/\beta}$-bounded.
It suffices that:
$$n_1 = \frac{\ln f - \frac{1}{\beta}\ln 2}{\ln\left(1-\frac{\lambda}{2}\right)} = O\left(\frac{1}{\lambda}\ln f\right).$$
If $2^{1/\beta} \le 1 + 2\delta/\beta$, then $O(\frac{1}{\lambda} \ln f)$ days suffice.

After this, by Lemma \ref{lem:conv-one-time-market-1}, after an additional $n_2$ days,
the market becomesa $1+(2\delta/\beta)$-bounded,
if $n_2$ satisfies the inequality $(2^{1/\beta})^{(1-\lambda\beta/(2\ln 2))^{n_2}}\leq 1+2\delta/\beta$.
It suffices that:
$$n_2 = \frac{\ln \beta + \ln\ln (1+2\delta/\beta) -\ln\ln 2}{\ln (1-\lambda\beta/(2\ln 2))} = O\left(\frac{1}{\lambda\beta}\left(\ln\frac{1}{\beta} + \ln\frac{1}{\delta}\right)\right).$$
The last equality holds as $\delta/\beta \leq 1$ and hence $\ln\ln (1+2\delta/\beta) = \ln \left(\delta/\beta\right) + O(1)$.

The sum $n_1+n_2$ bounds the number of days Phase 1 lasts.
\end{proof}

\noindent
{\bf Comment}.
If we wish to analyze the one-time market or the ongoing market without
taking account of the warehouses, then arbitrarily accurate prices can be achieved
in Phase 1, and the time till prices are $(1 + \eta)$-bounded,
for any $\eta$ is given by the bound in Theorem~\ref{th:phase-1-comps},
on replacing the term $\frac{1}{\beta}\log\frac{1}{\delta}$
by $\frac{1}{\beta}\log\frac{1}{\eta}$.

To apply this analysis of Phase 1 to other markets, it suffices
to identify conditions that ensure ${x}_1\ge f^{\beta}w_1$ when
$p_1 = p_1^*/f$, and ${x}_1 \le f^{-\beta}$ when $p_1 = f p_1^*$.

\paragraph{Phase 2}
Once the warehouse excesses may have a large impact on the price updates,
we can no longer demonstrate a smooth shrinkage of the term $(f-1)$.
Instead, we use an amortized analysis. We associate the following potential
$\phi_i$ with good $G_i$.
$$\phi_i := p_i [ \mbox{span}\{\bar{x}_i, x_i, \widetilde{w}_i\} -
c_1 \lambda (t-\tau_i) |\bar{x}_i - \widetilde{w}_i| + c_2 |\widetilde{w}_i - w_i| ],$$
where $\mbox{span}\{t_1,t_2,t_3\} = \max\{t_1,t_2,t_3\} - \min\{t_1,t_2,t_3\}$
and $1\geq c_1 > 0,\,c_2 > 1$ are suitably chosen constants.
We define $\phi := \sum_i \phi_i$.
The term $-c_1 \lambda (t-\tau_i) |\bar{x}_i - \widetilde{w}_i|$ ensures that $\phi$
decreases smoothly when no price update is occurring, as shown in the following
lemma.

\begin{lemma}\label{lem:PF-change-no-price-update}
Suppose that $4\kappa (1+c_2) \leq \lambda c_1 \leq 1/2$.
If $|\widetilde{w}_i-w_i|\leq 2\cdot\mbox{span}(x_i,\bar{x}_i,\widetilde{w}_i)$,
then $\frac{d\,\phi_i}{d\,t}\leq -\frac{\kappa(1+c_2)}{1+2c_2}\phi_i$ and otherwise
$\frac{d\,\phi_i}{d\,t}\leq -\frac{\kappa(c_2-1)}{2 c_2}\phi_i$,
at any time when no price update is occuring (to any $p_j$).
\end{lemma}

\begin{proof}To simplify the presentation of this proof,
let $K$ denote $\kappa (x_i-w_i)$ and let $S$ denote $\mbox{span}(x_i,\bar{x}_i,\widetilde{w}_i)$.

Note the following equalities: $\frac{d\,x_i}{d\,t} = \frac{d\,w_i}{d\,t} = 0$, $\frac{d\,\widetilde{w}_i}{d\,t} = -K$ and $\frac{d\,\bar{x}_i}{d\,t} = \frac{x_i - \bar{x}_i}{t-\tau_i}$. Then $\frac{d\,c_2|\widetilde{w}_i-w_i|}{d\,t} = -c_2 K\cdot \mbox{sign}(\widetilde{w}_i-w_i)$ and hence
\begin{eqnarray*}
\frac{d\,\phi_i}{d\,t} &=& p_i\left[\frac{d\,S}{d\,t} - c_1\lambda |\bar{x}_i-\widetilde{w}_i| -c_1\lambda (t-\tau_i) \frac{d\,|\bar{x}_i-\widetilde{w}_i|}{d\,t}\right. \\
&& \hspace{0.2in} -c_2 K\cdot \mbox{sign}(\widetilde{w}_i-w_i) \bigg]\\
&=& p_i\left[\frac{d\,S}{d\,t} - c_1\lambda (x_i-\widetilde{w}_i)\cdot\mbox{sign}(\bar{x}_i-\widetilde{w}_i)\right. \\
&& \hspace{0.2in} - c_1\lambda (t-\tau_i)K\cdot\mbox{sign}(\bar{x}_i-\widetilde{w}_i) \\
&& \hspace{0.2in} -c_2 K\cdot \mbox{sign}(\widetilde{w}_i-w_i)\bigg].
\end{eqnarray*}

Next by means of a case analysis, we show that
\begin{equation}\label{eq:universal-bound-pf-derivative}
\frac{d\,\phi_i}{d\,t} \leq p_i\left[|K| - c_1\lambda S -c_2 K\cdot \mbox{sign}(\widetilde{w}_i-w_i)\right].
\end{equation}
We show Case 1 in detail. Cases 2 and 3 are similar.

\smallskip

\noindent
{\bf Case 1:} $x_i\geq \bar{x}_i\geq \widetilde{w}_i$ or $\widetilde{w}_i\geq \bar{x}_i\geq x_i$. $\frac{d\,S}{d\,t} = K\cdot\mbox{sign}(x_i-\widetilde{w}_i)$.
\begin{eqnarray*}
\frac{d\,\phi_i}{d\,t} &=& p_i\left[(K - c_1\lambda (x_i-\widetilde{w}_i)- c_1\lambda (t-\tau_i)K)\mbox{sign}(x_i-\widetilde{w}_i)\right.\\
&& \hspace{0.2in} -c_2 K\cdot \mbox{sign}(\widetilde{w}_i-w_i)]\\
&=& p_i\left[K(1-c_1\lambda (t-\tau_i))\mbox{sign}(x_i-\widetilde{w}_i)\right.\\
&& \hspace{0.2in} -c_1\lambda |x_i-\widetilde{w}_i| -c_2 K\cdot \mbox{sign}(\widetilde{w}_i-w_i)]\\
&\leq & p_i\left[|K| - c_1\lambda S - c_2 K\cdot \mbox{sign}(\widetilde{w}_i-w_i)\right].
\end{eqnarray*}

\noindent
{\bf Case 2:} $x_i\geq\widetilde{w}_i\geq\bar{x}_i$ or $\bar{x}_i\geq\widetilde{w}_i\geq x_i$.

\smallskip

\noindent
{\bf Case 3:} $\widetilde{w}_i\geq x_i\geq \bar{x}_i$ or $\bar{x}_i\geq x_i\geq \widetilde{w}_i$.

\smallskip

Now we use the bound from \eqref{eq:universal-bound-pf-derivative} to obtain the bounds on the derivatives stated in the lemma.
There are two cases: $|\widetilde{w}_i-w_i|\leq 2S$ and $|\widetilde{w}_i-w_i| > 2S$.

\smallskip

\noindent
{\bf Case 1:}
$|\widetilde{w}_i-w_i|\leq 2S$. \\
Then $|x_i-w_i| \leq |x_i-\widetilde{w}_i| + |\widetilde{w}_i-w_i| \leq S + 2S = 3S$. And
\begin{eqnarray*}
\frac{d\,\phi_i}{d\,t} &\leq & -c_1\lambda p_i S + (1+c_2) p_i |K|\\
& \leq &  -c_1\lambda p_i S + 3(1+c_2) p_i \kappa S\\
&=& -(c_1\lambda - 3(1+c_2)\kappa) p_i S\\
& \leq &  -\frac{c_1\lambda - 3(1+c_2)\kappa}{1+2c_2} p_i (S+c_2|\widetilde{w}_i-w_i|)\\
&\leq & -\frac{c_1\lambda - 3(1+c_2)\kappa}{1+2c_2} \phi_i\\
& \leq & -\frac{\kappa (1+c_2)}{1+2c_2}\phi_i.
\end{eqnarray*}

\noindent
{\bf Case 2:}
$|\widetilde{w}_i-w_i| > 2S$.\\
Then $|\widetilde{w}_i-w_i| \leq |\widetilde{w}_i-x_i| + |x_i-w_i| \leq S + |x_i-w_i| < \frac{|\widetilde{w}_i-w_i|}{2} + |x_i-w_i|$ and hence $|\widetilde{w}_i-w_i| < 2 |x_i-w_i|$. Note that $\mbox{sign}(x_i-w_i) = \mbox{sign}(\widetilde{w}_i-w_i)$, so $-c_2 K\cdot \mbox{sign}(\widetilde{w}_i-w_i) = -c_2\kappa |x_i-w_i|$.
\begin{eqnarray*}
\frac{d\,\phi_i}{d\,t} &\leq & -c_1\lambda p_i S + \kappa p_i |x_i-w_i| - c_2\kappa p_i |x_i-w_i|\\
&=& -c_1\lambda p_i S -(c_2-1)\kappa p_i |x_i-w_i|\\
&<& -c_1\lambda p_i S -\frac{c_2-1}{2}\kappa p_i |\widetilde{w}_i-w_i|\\
&<& -\frac{c_2-1}{2c_2}\kappa p_i (c_2 S+ c_2|\widetilde{w}_i-w_i|)\\
&\leq & -\frac{\kappa(c_2-1)}{2 c_2}\phi_i.
\end{eqnarray*}
\end{proof}

The remaining task is to show that $\phi$ is non-increasing when a price update occurs.
This entails showing that the decrease to the term $p_i\cdot \mbox{span}\{\bar{x}_i, x_i, \widetilde{w}_i\}$
is at least as large as the increase to the term $p_i c_2 |\widetilde{w}_i - w_i|$ plus the value of the
term $p_i c_1 \lambda (t-\tau_i) |\bar{x}_i - \widetilde{w}_i|$, which gets reset to 0.


\begin{lemma}\label{lem:ongoing-price-change-1}
Let $\beta = 2\alpha - \gamma$ and
suppose that $\beta > 0$ and the following conditions hold:
\begin{enumerate}
\item $f \le (1-2\delta)^{-1/\beta} \le (2-\delta)^{1/\gamma}$ since the last price update to $p_i$;
\item $\bar{\alpha} + c_1 + c_2 \delta \leq 1-\delta$;
\item $(1+\delta+c_1+c_2\delta)\lambda \leq 1$,
\end{enumerate}
where $\bar{\alpha} := 2(1-\alpha)(1-2\delta)^{-\gamma / \beta}
\left(1+\frac{\alpha \lambda(1+\delta)}{2(1-\lambda(1+\delta))}\right)$.
Then, when a price $p_i$ is updated using rule \eqref{eq:tat-rule-warehouse},
the value of $\phi$ stays the same or decreases.
\end{lemma}

If $\delta$ and $\lambda$ are small, then
\begin{eqnarray*}
\bar{\alpha} &=& (2-2\alpha)\left(1+O\left(\frac{\gamma\delta}{\beta}\right)\right)\left(1+O(\lambda)\right) \\
&=& 2-2\alpha + O\left(\frac{\gamma\delta}{\beta}\right) + O(\lambda),
\end{eqnarray*}
and Condition (2) becomes $2-2\alpha + c_1 + O(\delta(1+\gamma/\beta)) + O(\lambda) \leq 1$,
which is satisfied on setting $c_1,\lambda,\delta(1+\gamma/\beta) = O(2\alpha-1)$.
The third condition is then satisfied by having $\lambda = O(1)$.
More precise bounds are given later.

To demonstrate a continued convergence of the prices during Phase 2, we need
to relate the prices to the potential $\phi$.
We show the following bound.

\begin{theorem}
\label{th:phase2-comp}
Suppose that the conditions in Lemmas \ref{lem:PF-change-no-price-update} and \ref{lem:ongoing-price-change-1} hold. Let $M= \sum_j b_j$ be the daily supply of money to all the buyers.
Then, in Phase 2, the prices become $(1 + \eta)$-bounded after
$O\left(\frac{1}{\kappa} \log \frac{M}{\eta \min_i w_i p_i^*}\right)$
days.
\end{theorem}
\begin{proof}
During Phase 2, $p_i \le 2 p_i^*$ and $x_i \le 2w_i$.
Consequently,
$\phi = O(\sum_i p^*_i w_i) = O(M)$.
Once $\phi$ has shrunk to $\eta \min_i p_i^*w_i$,
we know that all prices are $(1 + \eta)$-bounded.
Finally, Lemmas~\ref{lem:PF-change-no-price-update} and~\ref{lem:ongoing-price-change-1}
imply that $\phi$ shrinks by a $(1-\Theta(\kappa))$ factor each day.
\end{proof}

If the updates in Phase 2 start with an initial value for the potential
of $\phi_{\mbox{\tiny I}} \ll M$,
then in the bound on the number of days one can replace $M$ with $\phi_{\mbox{\tiny I}}$.

Summing the bounds from Theorems \ref{th:phase-1-comps} (note that $\delta \leq \beta$) and \ref{th:phase2-comp}, yields
Theorem~\ref{th:comp-overall}, modulo showing that Equations~\eqref{eqn:kappa-cond1}--\eqref{eqn:lambda-cond1} suffice to ensure
the conditions in Theorems \ref{th:phase-1-comps} and \ref{th:phase2-comp}.

\subsection{Bounds on the Warehouse Sizes}

\paragraph{Phase 1}
Recall that $f_{\mbox{\tiny I}}$ is the initial value of $f$.
Define $d(f) = \max_i x_i/w_i$ when the prices are $f$-bounded.
In a market of complementary goods, $d(f) \le f^{\gamma}$.
We can show:

\begin{lemma}
\label{lem:phase1-war-bdd}
In Phase 1, the total net change to $v_i$ is bounded by
$O(\frac{w_i}{\lambda} d(f_{\mbox{\tiny I}}) + \frac{w_i}{\lambda\beta} d(2) \log \frac{\beta}{\delta})$.
\end{lemma}

\paragraph{Phase 2}
Because Phase 2 may last $\chi(1/\kappa)$ days,
we cannot simply use a bound on its duration to bound the capacity needed for warehouse $i$,
for its capacity is $O(\frac{w_i}{\kappa})$,
which could be smaller than the bound based on the duration of Phase 2.

Instead, we observe that in Phase 2 the price adjustments are always
strictly within the bounds of $1 \pm \lambda\Delta t$,
where $\Delta t$ is the time since the previous update to $p_i$.
If $v_i \le \chi_i/2 - bw_i$, then an update of $p_i$ by a factor $1 -\lambda \mu \Delta t$,
implies that $w_i - \bar{x}_i \ge (\mu w_i  + \kappa b w_i)\Delta t$,
and $(w_i - \bar{x}_i)\Delta t$ is exactly the amount by which $v_i$ decreases between these
two price updates.
As the prices are $((1 - 2\delta)^{-\beta})$-bounded, the difference between
the price increases and decreases is bounded, and consequently, over time
the change to the warehouse stock will be dominated by the sum of the
$\kappa b w_i \Delta t$ terms.
This is made precise in the following lemma
(an analogous result applies if $v_i \ge \chi_i/2 + bw_i)$.

\begin{lemma}
\label{lem:war-too-empt-impr}
Let $a_1, a_2, k > 0$.
Suppose that
$v_i \leq \vis - a_1 w_i $ and that $\kappa a_1 \ge 4 \lambda^2$.
Let $\tau$ be the time of a price update of $p_i$ to
$p_{i,1}$.
Suppose that henceforth $ p_{i} \leq e^{\bar{f}} p_{i,1} $
for some $\bar{f} \geq 0$.
If $k \ge \frac{2}{\kappa a_1}(\bar{f} + a_2)$, then
by time $\tau + (k+1)$ the warehouse
stock will have increased to more than $\vis - a_1 w_i $,
or by at least $a_2 w_i$, whichever is the lesser increase.
\end{lemma}
\begin{proof}
Suppose that $v_i \leq \vis - a_1 w_i$ throughout (or the result holds trivially).

Each price change by a multiplicative $(1 + \mu \Delta t)$ is associated
with a target excess demand $\bar{z}_i = \bar{x}_i -w_i - \kappa (v_i -v_i^*)$,
where $\zbi = \mu w_i$.
Furthermore, the increase to the warehouse stock since the previous price
update is exactly $-(\xbi - w_i)\Delta t = [- \mu w_i - \kappa (v_i -v_i^*)]\Delta t \ge
(-\mu + \kappa a_1)w_i  \Delta t$.

Note that $1 + x \geq e^{x -2 x^2}$ for $|x| \leq \frac 12$.
Thus $1 + \mu \Delta t \ge e^{\mu \Delta t - 2\lambda^2 \Delta t}$
(recall that all price changes are bounded by $1 \pm \lambda \Delta t$).

Suppose that over the next $k$ days there are $l-1$ price changes;
let the next $l$ price changes be by $1 + \mu_1 \Delta t_1, 1 + \mu_2 \Delta t_2,
\cdots, 1 + \mu_l \Delta t_l$.
Note that the total price change satisfies
$e^{\bar{f}} \ge \Pi_{1\le i \le l} (1 + \mu_i \Delta t_i)
\ge e^{\sum_{1\le i \le l} (\mu_i \Delta t_i - 2 \lambda^2 \Delta t_i))}$.
Thus $\sum_{1\le i \le  l} \Delta t_i(\mu_i   - 2 \lambda^2 ) \le \bar{f}$.

We conclude that when the $l$-th price change occurs,
the warehouse stock will have increased by at least
$\sum_{1\le i \le l}(-\mu_i  + \kappa a_1)w_i \Delta t_i
\ge -  \bar{f} + k(-2 \lambda^2 + \kappa a_1)w_i
\ge (-  \bar{f} + \frac 12 k\kappa a_1)w_i$,
both inequalities following because $4\lambda^2 \le \kappa a_1$.
If $k \ge \frac{2}{\kappa a_1}(\bar{f} + a_2)$, then
the warehouse stock increases by at least $a_2 w_i$.\end{proof}

\noindent
{\bf Comment.}
The relationship between the change in capacity and the size of the price update
is crucial in proving this lemma, and this depends on having the factor $\Delta t$
in the price update rule.

\smallskip

To complete the analysis of Phase 2, we view each warehouse as having
8 equal sized zones of fullness, with the goal being to bring
the warehouse into its central four zones. The role of the outer zones is to provide a buffer
to cope with initial price imbalances.

\begin{definition}
\label{def:wrhs-zones}
The four zones above the half way target are called the \emph{high} zones, and the other four
are the \emph{low} zones.
Going from the center outward, the zones are called the \emph{central} zone, the \emph{inner buffer},
the \emph{middle buffer}, and the \emph{outer buffer}.
The warehouse is said to be \emph{safe} if it is in one of its central zones or one of its
inner buffers.
\end{definition}

Let $D(f_{\mbox{\tiny I}})$ bound the duration of Phase 1 and
let $v(f_{\mbox{\tiny I}})$ be chosen so that
$v(f_{\mbox{\tiny I}}) w_i$ bounds the change to $v_i$, for all $i$, during Phase 1.
We gave a bound on $v(f_{\mbox{\tiny I}})$ in Lemma \ref{lem:phase1-war-bdd}.

We will assume that the ratios $\chi_i/w_i$ are all the same, i.e.\ that
every warehouse can store the same maximum number of days supply.
This will be without loss of generality, for if the smallest
warehouse can store only $2d$ days supply, Theorem~\ref{lem:good-wrhs}
in effect shows that every
warehouse remains with a stock within $dw_i$ of $\chi_i/2$.
An alternative approach is to suppose that each seller $S_i$ has
a separate parameter $\kappa_i$ (replacing $\kappa$).
The only effect on the analysis is that the convergence rate
is now controlled by $\kappa = \min_i \kappa_i$.

To prove Theorem~\ref{lem:good-wrhs} it will suffice that the following conditions hold.
 for all $i$:
\begin{enumerate}
\item
$\chi_i \ge \frac{512}{\beta} w_i$ and
$\chi_i \ge 8 v(f_{\mbox{\tiny I}}) w_i $.
\item
$\delta = \frac{\kappa \chi_i}{2 w_i}$.
\item
$\lambda^2 \le \frac {\kappa \chi_i} {32 w_i}$.
\end{enumerate}

\noindent
{\bf Comment.}
We note that were the price update rule to have the form $p_i' \leftarrow p_i e^{\lambda \min\{1,\bar{z}/w_i\}\Delta t}$ rather than
$p_i' \leftarrow p_i (1+\lambda \min\{1,\bar{z}/w_i\}\Delta t)$ then the constraint (3) in Theorem \ref{lem:good-wrhs} would not be needed 
(this constraint comes from setting $a_1$ in Lemma \ref{lem:war-too-empt-impr} to the width of a zone).
We call this alternate rule the \emph{exponential price update rule}.
However, we prefer the form of the rule we have specified as it strikes us as being simpler and hence more natural.

\subsection{Condition Summary}
\label{sec:ciond-sum}

Lemma \ref{lem:ongoing-price-change-1} and Theorem \ref{lem:good-wrhs} require several constraints on the parameters 
$\kappa, \delta, \lambda, c_1, c_2$.
We can unwind these conditions to show how these parameters depend on the market parameters $\alpha,\gamma$ and $\beta$.

Let $r = \chi_i / w_i$. Then the conditions can be satisfied when 
Equations~\eqref{eqn:kappa-cond1} and \eqref{eqn:lambda-cond1} hold. 
Note that $r$ needs to be sufficiently large, or in other words $\chi_i$ for every $i$ needs to be sufficiently large, to ensure that there is a choice of $\lambda$ which satisfies both the upper and lower bounds.
Further note that the term $\sqrt{\frac{\kappa\tau}{32}}$,
which is due to Constraint (3), would not be needed
were we to use the exponential price update rule.

\section{Markets with Mixtures of Substitutes and Complements}\label{sec:mixed}

For the markets with mixtures of substitutes and complements, we defined Adverse Market Elasticity and made Assumption \ref{ass:sp-trans} in Section \ref{sect:result-mix}. Note that for the case that all the goods are complements, $\beta$ as defined in Definition \ref{def:den-elas} equals $2\alpha-\gamma$.

We can then show that Theorem~\ref{th:phase-1-comps} applies here too.
We can also show the following results,
analogs of Lemma \ref{lem:ongoing-price-change-1}, and Theorems \ref{th:phase2-comp} and \ref{th:comp-overall}.

\begin{lemma}\label{lem:ongoing-price-change-2}
Suppose Assumption 1 holds and $\beta$, as defined in Definition \ref{def:den-elas}, satisfies $\beta>0$. Suppose that the following conditions hold:
\begin{enumerate}
\item $f \leq (1-2\delta)^{-1/\beta} \leq (2-\delta)^{1/(2E-\beta)}$ since the last price update to $p_i$;
\item $2\alpha '(1-2\delta)^{-(2E-\beta)/\beta} + c_1 + c_2 \delta \leq 1-\delta$;
\item $\left(2\alpha ''(1-2\delta)^{-(2E-\beta)/\beta} + 1 + \delta + c_1 + c_2 \delta\right)\lambda \leq 1$,
\end{enumerate}
where $\alpha '' := \alpha' + 2(E-1)$.
Then, when a price $p_i$ is updated using rule \eqref{eq:tat-rule-warehouse}, the value of $\phi$ stays the same or decreases.
\end{lemma}

\begin{theorem}
\label{th:phase2-comp-multi}
Suppose that the conditions in Lemmas \ref{lem:PF-change-no-price-update} and \ref{lem:ongoing-price-change-2} hold.
Let $M= \sum_j b_j$ be the daily supply of money to all the buyers.
Then, in Phase 2, the prices become $(1 + \eta)$-bounded after
$O\left(\frac{1}{\kappa} \log \frac{M}{\eta \min_i w_i p_i^*}\right)$
days.
\end{theorem}

Theorem~\ref{th:multi-overall} follows on summing the bounds from Theorems~\ref{th:phase-1-comps} and~\ref{th:phase2-comp-multi},
and on showing that Equations~\eqref{eqn:kappa-cond1}--\eqref{eqn:lambda-cond1}
imply the constraints in Lemmas~\ref{lem:PF-change-no-price-update} and \ref{lem:ongoing-price-change-2}.
Theorem~\ref{lem:good-wrhs} also continues to apply unchanged.
Here $d(f) \le f^{2E - \beta}$.

\subsection{Example Scenario: 2-Level Nested CES Type Utilities}

We will use index $i$ to denote a good, $\overline{G_i}$ to denote the group containing good $i$, index $j$ to a denote a good in $\overline{G_i}$ (but not good $i$) and index $k$ to denote a good in a group other than $\overline{G_i}$. Denote the spending on all the goods in the group $\overline{G_i}$ by $s_{\overline{G_i}}$ and the total income of the buyer by $b$. Keller \cite{Keller1976} derived the following formulae:
\begin{eqnarray*}
\frac{\partial x_i/\partial p_i}{x_i / p_i} &=& -\frac{1}{1-\rho_{\overline{G_i}}}\left(1-\frac{s_i}{s_{\overline{G_i}}}\right) - \frac{1}{1-\rho}\left(\frac{s_i}{s_{\overline{G_i}}} - \frac{s_i}{b}\right) - \frac{s_i}{b}\\
\\
\frac{\partial x_i/\partial p_j}{x_i / p_j} &=& \frac{s_j}{b} \left( \frac{1}{1-\rho_{\overline{G_i}}} \frac{b}{s_{\overline{G_i}}} - \frac{1}{1-\rho}\left(\frac{b}{s_{\overline{G_i}}}-1\right) - 1 \right)\\
\\
\frac{\partial s_k}{\partial p_i} &=& \frac{s_k}{b} \frac{\rho}{1-\rho} x_i.
\end{eqnarray*}

As $1>\rho_{\overline{G}}>0$, $\rho<0$ and $b\geq s_{\overline{G}}$, $\frac{\partial x_i}{\partial p_j} \geq 0$ and $\frac{\partial x_i/\partial p_i}{x_i / p_i}\geq -\frac{1}{1-\rho_{\overline{G}}}$;
i.e.\ every pair of goods in the same group are substitutes and $E=\max_{\overline{G}} \frac{1}{1-\rho_{\overline{G}}}$.

As $\rho<0$, $\frac{\partial s_k}{\partial p_i} < 0$, which is equivalent to $\frac{\partial x_k}{\partial p_i} < 0$;
i.e.\ two goods in different groups are complements.

To compute $\beta$, we note that when $p_i$ changes by a factor $t$, the smallest change in demand occurs
if the prices for its substitutes, namely, the goods in its group, all also change by $t$, while the prices for 
its complements, namely all the other goods, change
by a factor $1/t$.
As $\rho<0$, $\frac{\partial x_i/\partial p_i}{x_i/p_i} + \sum_{j\in \overline{G},j\neq i} \frac{\partial x_i/\partial p_j}{x_i/p_j} = -\frac{1}{1-\rho}-\frac{s_{\overline{G}}}{b}\left(1-\frac{1}{1-\rho}\right)\leq -\frac{1}{1-\rho}$. When the prices of all goods are raised by a factor $t>1$ and then the prices of all goods in $\overline{G}$ are reduced by a factor $1/t^2$, $x_i' \geq x_i t^{2/(1-\rho)-1}$;
when the prices of all goods are reduced by a factor $t<1$ and then the prices of all goods in $\overline{G}$ are raised by a factor $1/t^2$, $x_i' \leq x_i t^{2/(1-\rho)-1}$. Thus $\beta = \frac{2}{1-\rho}-1$.

Finally, note that $\sum_k \frac{\partial s_k}{\partial p_i} = \frac{\sum_k s_k}{b} \frac{\rho}{1-\rho} x_i$ and $|\Delta S_c| = \sum_k |\Delta s_k|$. When $p_i$ is raised, $x_i'\leq x_i$,
and hence $|\Delta S_c| \leq \frac{-\rho}{1-\rho} x_i |\Delta p_i|$; when $p_i$ is reduced, it is reduced by a factor of $t\geq(1-\lambda)$. 
As $x_i'\leq x_i (1-\lambda)^{-E}$, $|\Delta S_c|\leq \frac{-\rho}{1-\rho}(1-\lambda)^{-E} x_i |\Delta p_i|$. Thus Assumption \ref{ass:sp-trans} is satisfied with $\alpha' = \frac{-\rho}{1-\rho}(1-\lambda)^{-E}$. Hence the bounds from Theorems \ref{th:multi-overall} and~\ref{lem:good-wrhs} apply.

\section{Acknowledgments}

This work is supported by the National Science Foundation, under
NSF grant CCF-0830516.

The second author thanks Lisa Fleischer for helpful comments regarding
the complementary goods case.

\bibliographystyle{ieeetr}
\bibliography{markets}

\newpage
\appendix

\section{Potential Function Lemmas}

We complete the proof of the missing cases for Lemma \ref{lem:PF-change-no-price-update}.

\smallskip

\begin{rlemma}{lem:PF-change-no-price-update}
Suppose that $4\kappa (1+c_2) \leq \lambda c_1 \leq 1/2$.
If $|\widetilde{w}_i-w_i|\leq 2\cdot\mbox{span}(x_i,\bar{x}_i,\widetilde{w}_i)$,
then $\frac{d\,\phi_i}{d\,t}\leq -\frac{\kappa(1+c_2)}{1+2c_2}\phi_i$ and otherwise
$\frac{d\,\phi_i}{d\,t}\leq -\frac{\kappa(c_2-1)}{2 c_2}\phi_i$,
at any time when no price update is occuring (to any $p_j$).
\end{rlemma}

\begin{proof}To simplify the presentation of this proof, let $K$ denote $\kappa (x_i-w_i)$ and let $S$ denote $\mbox{span}(x_i,\bar{x}_i,\widetilde{w}_i)$.
Here we show the details for Cases 2 and 3, which were deferred from the main paper.

\smallskip

\noindent\textbf{Case 2:} $x_i\geq\widetilde{w}_i\geq\bar{x}_i$ or $\bar{x}_i\geq\widetilde{w}_i\geq x_i$. $\frac{d\,S}{d\,t} = \frac{\bar{x}_i-x_i}{t-\tau_i}\cdot\mbox{sign}(x_i-\bar{x}_i)$.
\begin{eqnarray*}
&& \frac{d\,\phi_i}{d\,t} \\
&=& p_i\left[\left(\frac{\bar{x}_i-x_i}{t-\tau_i} + c_1\lambda (x_i-\widetilde{w}_i) + c_1\lambda (t-\tau_i) K \right)\mbox{sign}(x_i-\bar{x}_i) \right.\\
&& \hspace{0.2in} -c_2 K\cdot \mbox{sign}(\widetilde{w}_i-w_i)\bigg]\\
&\leq & p_i\left[-|\bar{x}_i-x_i|+c_1\lambda |x_i-\widetilde{w}_i| + |K| -c_2 K\cdot \mbox{sign}(\widetilde{w}_i-w_i)\right]\\
&\leq & p_i\left[|K| + (c_1\lambda-1) |\bar{x}_i-x_i| -c_2 K\cdot \mbox{sign}(\widetilde{w}_i-w_i)\right]\\
&\leq & p_i\left[|K| - c_1\lambda S -c_2 K\cdot \mbox{sign}(\widetilde{w}_i-w_i)\right],~\mbox{as $\lambda c_1 \le \frac 12$}.
\end{eqnarray*}

\noindent\textbf{Case 3:} $\widetilde{w}_i\geq x_i\geq \bar{x}_i$ or $\bar{x}_i\geq x_i\geq \widetilde{w}_i$. $\frac{d\,S}{d\,t} = \left(\frac{x_i-\bar{x}_i}{t-\tau_i} + K\right)\cdot\mbox{sign}(\bar{x}_i-\widetilde{w}_i)$.

\begin{eqnarray*}
&& \frac{d\,\phi_i}{d\,t}\\
&=& p_i\left[\left( \frac{\bar{x}_i-x_i}{t-\tau_i} - K + c_1\lambda (x_i-\widetilde{w}_i) + c_1\lambda (t-\tau_i) K\right) \mbox{sign}(\widetilde{w}_i - \bar{x}_i) \right.\\
&& \hspace{0.2in} -c_2 K\cdot \mbox{sign}(\widetilde{w}_i-w_i)\bigg]\\
&\leq & p_i\left[-|\bar{x}_i-x_i| - c_1\lambda |x_i-\widetilde{w}_i|\right.\\
&& \hspace{0.2in} - K(1-c_1\lambda(t-\tau_i))\mbox{sign}(\widetilde{w}_i - \bar{x}_i) -c_2 K\cdot \mbox{sign}(\widetilde{w}_i-w_i) ]\\
&\leq & p_i\left[-c_1\lambda |\bar{x}_i-x_i| - c_1\lambda |x_i-\widetilde{w}_i| + |K| -c_2 K\cdot \mbox{sign}(\widetilde{w}_i-w_i)\right]\\
&=& p_i\left[|K| - c_1\lambda S -c_2 K\cdot \mbox{sign}(\widetilde{w}_i-w_i)\right].
\end{eqnarray*}
\end{proof}

The following lemma provides an upper bound on the change to the potential function when there is a price update.
Subsequently, this lemma will be used to show that at a price update, the potential function stays the same or decreases under suitable conditions. Recall that $s_j$ denotes the spending on good $j$.

\begin{lemma}\label{lem:change-PF}
Suppose $p_i$ is updated. Let $S_{inc} := \sum_{j\neq i,\,\Delta s_j > 0} |\Delta s_j|$ and $S_{dec} := \sum_{j\neq i,\,\Delta s_j<0} |\Delta s_j|$.
\begin{enumerate}
\item If $\mbox{sign}(x_i-\widetilde{w}_i)$ is not flipped and $x_i$ moves towards $\widetilde{w}_i$, the change to $\phi$ is at most
\begin{eqnarray*}
-\widetilde{w}_i |\Delta p_i| + \mbox{sign}(\Delta p_i)\cdot\Delta s_i + S_{inc} + S_{dec}\\
+ c_1 \lambda p_i |\bar{x}_i - \widetilde{w}_i|(t-\tau_i) + c_2 \delta w_i |\Delta p_i|.
\end{eqnarray*}

\item If $\mbox{sign}(x_i-\widetilde{w}_i)$ is not flipped and $x_i$ moves away from $\widetilde{w}_i$, or if $\mbox{sign}(x_i-\widetilde{w}_i)$ is flipped, the change to $\phi$ is at most
\begin{eqnarray*}
-p_i |\bar{x}_i - \widetilde{w}_i| + \widetilde{w}_i |\Delta p_i| - \mbox{sign}(\Delta p_i)\cdot\Delta s_i + S_{inc} + S_{dec} \\
+ c_1 \lambda p_i |\bar{x}_i - \widetilde{w}_i|(t-\tau_i) + c_2 \delta w_i |\Delta p_i|.
\end{eqnarray*}

\end{enumerate}
\end{lemma}
\begin{proof} Let $p_i'$, $x_i'$ and $s_i'=p_i'x_i'$ denote the price of good $G_i$, the demand for good $G_i$ and the spending on good $G_i$ after  the price update respectively. We separate the proof into three cases.

\medskip

\noindent\textbf{Case 1:} $\mbox{sign}(x_i-\widetilde{w}_i)$ is not flipped and $x_i$ moves towards $\widetilde{w}_i$.

As $x_i$ moves towards $\widetilde{w}_i$, following the update, $\mbox{sign}(\Delta p_i) = \mbox{sign}(\bar{x}_i-\widetilde{w}_i) = \mbox{sign}(x_i-\widetilde{w}_i) = \mbox{sign}(x_i'-\widetilde{w}_i)$.

Consider the term $p_i\cdot\mbox{span}\{\bar{x}_i, x_i, \widetilde{w}_i\}$.
Before the update to $p_i$, it equals
$$p_i\cdot\mbox{span}\{\bar{x}_i, x_i, \widetilde{w}_i\} \geq p_i |x_i - \widetilde{w}_i| = (s_i - p_i \widetilde{w}_i)\cdot\mbox{sign}(\Delta p_i).$$
After the update,
\begin{eqnarray*}
p_i\cdot\mbox{span}\{\bar{x}_i, x_i, \widetilde{w}_i\} &=& (p_i + \Delta p_i)|x_i' - \widetilde{w}_i|\\
&=& (s_i' -p_i \widetilde{w}_i - \widetilde{w}_i\Delta p_i)\cdot\mbox{sign}(\Delta p_i).
\end{eqnarray*}

Hence, the change to the term following the update is at most
$$(s_i' - s_i - \widetilde{w}_i\Delta p_i)\cdot\mbox{sign}(\Delta p_i) = \mbox{sign}(\Delta p_i)\cdot\Delta s_i - \widetilde{w}_i |\Delta p_i|.$$

For the terms $- p_i c_1 \lambda (t-\tau_i) |\bar{x}_i - \widetilde{w}_i| + c_2 p_i |\widetilde{w}_i - w_i|$, an update on $p_i$ resets $\tau_i$ to $t$ and $|\widetilde{w}_i - w_i|\leq \delta w_i$. Hence the change to these two terms is at most $c_1 \lambda p_i |\bar{x}_i - \widetilde{w}_i| (t-\tau_i) + c_2 \delta w_i |\Delta p_i|$.

For any other good $j$, the terms $- p_j c_1 \lambda (t-\tau_j) |\bar{x}_j - \widetilde{w}_j| + c_2 p_j |\widetilde{w}_j - w_j|$ do not change, and the term $p_j\cdot\mbox{span}\{\bar{x}_j, x_j, \widetilde{w}_j\}$ changes, but by at most $\Delta s_j$. In the worst case, the change of this term, summing over all $j$, is at most $S_{inc} + S_{dec}$.

\medskip

\noindent\textbf{Case 2:} $\mbox{sign}(x_i-\widetilde{w}_i)$ is not flipped and $x_i$ moves away from $\widetilde{w}_i$.

As $x_i$ moves away from $\widetilde{w}_i$, following the update, $\mbox{sign}(\Delta p_i) = \mbox{sign}(\bar{x}_i-\widetilde{w}_i) \neq \mbox{sign}(x_i-\widetilde{w}_i) = \mbox{sign}(x_i'-\widetilde{w}_i)$.

Consider the term $p_i\cdot\mbox{span}\{\bar{x}_i, x_i, \widetilde{w}_i\}$. Before the update to $p_i$, it equals
\begin{eqnarray*}
p_i |x_i - \bar{x}_i| &=& p_i |x_i - \widetilde{w}_i| + p_i |\widetilde{w}_i - \bar{x}_i|\\
&=& (s_i - p_i \widetilde{w}_i)\cdot(-\mbox{sign}(\Delta p_i)) + p_i |\bar{x}_i - \widetilde{w}_i|.
\end{eqnarray*}

After the update,
$$(p_i+\Delta p_i)|x_i' - \widetilde{w}_i| = (s_i' - p_i \widetilde{w}_i - \widetilde{w}_i \Delta p_i)\cdot(-\mbox{sign}(\Delta p_i)).$$
Hence the change to the term following the update is at most
$$-\mbox{sign}(\Delta p_i) \Delta s_i + \widetilde{w}_i |\Delta p_i| -  p_i |\widetilde{w}_i - \bar{x}_i|.$$

As in Case 1, there are further changes, but bounded above by $S_{inc} + S_{dec} + c_1 \lambda p_i |\bar{x}_i - \widetilde{w}_i| (t-\tau_i) + c_2 \delta w_i |\Delta p_i|$.

\medskip

\noindent\textbf{Case 3:} $\mbox{sign}(x_i-\widetilde{w}_i)$ is flipped.

As $\mbox{sign}(x_i-\widetilde{w}_i)$ is flipped, $x_i$ moves toward $\widetilde{w}_i$ initially, hence $\mbox{sign}(\Delta p_i) = \mbox{sign}(\bar{x}_i-\widetilde{w}_i) = \mbox{sign}(x_i-\widetilde{w}_i) \neq \mbox{sign}(x_i'-\widetilde{w}_i)$.

Let $\widetilde{x}_i = \argmax_{x\in\{x_i,\bar{x}_i\}} |x-\widetilde{w}_i|$. Consider the term $p_i\cdot\mbox{span}\{\bar{x}_i, x_i, \widetilde{w}_i\}$. Before the update to $p_i$, it equals
$$p_i |\widetilde{x}_i - \widetilde{w}_i| = (p_i \widetilde{x}_i - p_i \widetilde{w}_i)\cdot\mbox{sign}(\Delta p_i).$$
After the update, it equals
$$(p_i+\Delta p_i)|x_i' - \widetilde{w}_i| = (s_i' - p_i \widetilde{w}_i - \widetilde{w}_i \Delta p_i)\cdot(-\mbox{sign}(\Delta p_i)).$$
Hence the change to the term following the update is at most
\begin{eqnarray*}
&& \widetilde{w}_i |\Delta p_i| - (\Delta s_i + s_i - p_i \widetilde{w}_i + p_i \widetilde{x}_i - p_i \widetilde{w}_i)\cdot\mbox{sign}(\Delta p_i)\\
&=& \widetilde{w}_i |\Delta p_i| - \mbox{sign}(\Delta p_i)\Delta s_i - p_i |x_i-\widetilde{w}_i| - p_i |\widetilde{x}_i-\widetilde{w}_i|.
\end{eqnarray*}

As $- p_i |x_i-\widetilde{w}_i|\leq 0$ and $- p_i |\widetilde{x}_i-\widetilde{w}_i|\leq - p_i |\bar{x}_i-\widetilde{w}_i|$, we obtain the same upper bound on the term $p_i\cdot\mbox{span}\{\bar{x}_i, x_i, \widetilde{w}_i\}$ as in Case 2. The rest of the argument is the same as in Case 2.

\end{proof}

\section{Markets with Complementary Goods}

The following lemma states several inequalities we will use. They can be proved by simple arithmetic/calculus.

\begin{lemma}
\begin{enumerate}[(a)]\label{lem:calculus-results}
\item If $0\leq\ep\leq 1$ and $0\leq x\leq 1$, then $(1+\ep)^x-1 \leq\ep x$.
\item If $0\leq\ep\leq 1$ and $0\leq x\leq 1$, then $1-(1-\ep)^x \leq \left(1+\frac{\ep}{2(1-\ep)}\right)\ep x$.
\item If $E\geq 1$, $\ep\geq 0$ and $r := \max\left\{\frac{E\ep}{2},\ep\right\} < 1$, then $(1-\ep)^{1-E}-1 \leq \frac{E-1}{1-r}\ep$.
\item If $E\geq 1$ and $0\leq\ep\leq 1$, then $1-(1+\ep)^{1-E} \leq (E-1)\ep$.
\item If $x\geq 1$ and $\ep\geq 0$, then $(1-\ep)^{-x} \leq 1+\frac{x}{1-\ep x}\ep$.
\end{enumerate}
\end{lemma}

\begin{rlemma}{lem:ongoing-price-change-1}
Let $\beta = 2\alpha - \gamma$ and
suppose that $\beta  > 0$ and the following conditions hold:
\begin{enumerate}
\item $f \le (1-2\delta)^{-1/\beta} \le (2-\delta)^{1/\gamma}$ since the last price update to $p_i$;
\item $\bar{\alpha} + c_1 + c_2 \delta \leq 1-\delta$;
\item $(1+\delta+c_1+c_2\delta)\lambda \leq 1$,
\end{enumerate}
where $\bar{\alpha} := 2(1-\alpha)(1-2\delta)^{-\gamma / \beta}
\left(1+\frac{\alpha \lambda(1+\delta)}{2(1-\lambda(1+\delta))}\right)$.
Then, when a price $p_i$ is updated using rule \eqref{eq:tat-rule-warehouse},
the value of $\phi$ stays the same or decreases.
\end{rlemma}

\begin{proof} The first condition is used with Lemma \ref{lem:complementary-demand-bound-prelim}(b) to ensure that $\bar{x}_i \leq (2-\delta) w_i$, which implies $\frac{\bar{z}_i}{w_i}\leq 1$.
Then, by price update rule \eqref{eq:tat-rule-warehouse}, $|\Delta p_i|=\lambda p_i |\bar{x}_i - \widetilde{w}_i| / w_i$.

\medskip

\noindent\textbf{Step 1:} This step shows that the amount of spending transferred due to a price change is bounded by $\bar{\alpha} w_i |\Delta p_i|$.

By the first condition and Lemma~\ref{lem:complementary-demand-bound-prelim}(b), $x_i \leq (1-2\delta)^{-\gamma/\beta} w_i$.
Hence by definition of $\bar{\alpha}$,
$2(1-\alpha)\left(1+\frac{\lambda(1+\delta)}{2(1-\lambda(1+\delta))}\right) x_i \leq \bar{\alpha} w_i$.

\noindent\textbf{Case 1(a):} Price $p_i$ is increased to $t p_i$, where $t>1$, i.e.\ $\Delta p_i = (t-1) p_i$.

By Lemma~\ref{lem:complementary-demand-bound-prelim}(c),
the spending increase on $G_i$ due to this price
increase is at most $(tp_i)\left(\frac{x_i}{t^{\alpha}}\right) - x_i p_i = (t^{1 - \alpha} -1)x_i p_i.$

By Lemma \ref{lem:calculus-results}(a), $t^{1-\alpha}-1 \leq (1-\alpha)(t-1)$. Hence $2(t^{1-\alpha}-1) p_i x_i$, which is twice the upper bound on the spending drawn from other goods due to the price increase, satisfies
\begin{eqnarray*}
2(t^{1-\alpha}-1) p_i x_i &\leq & 2(1-\alpha) (1-2\delta)^{-\gamma/\beta}(t-1) p_i w_i \\
&\leq & \bar{\alpha} w_i \left(1+\frac{\lambda(1+\delta)}{2(1-\lambda(1+\delta))}\right)^{-1} |\Delta p_i| \\
&\leq & \bar{\alpha} w_i |\Delta p_i|.
\end{eqnarray*}

\noindent\textbf{Case 1(b):} Price $p_i$ is reduced to $t p_i$, where $t<1$.

By Lemma~\ref{lem:complementary-demand-bound-prelim}(d),
the spending decrease on $G_i$ due to this price decrease is at most
$x_i p_i - (tp_i)\left(\frac{x_i}{t^{\alpha}}\right) = (1-t^{1-\alpha})x_i p_i.$

By price update rule \eqref{eq:tat-rule-warehouse}, $1 > t \geq 1-\lambda(1+\delta)$. By Lemma \ref{lem:calculus-results}(b), $(1-t^{1-\alpha}) \leq \left(1+\frac{\lambda(1+\delta)}{2(1-\lambda(1+\delta))}\right) (1-\alpha)(1-t)$. Hence $2(1-t^{1-\alpha}) p_i x_i$, which is twice the upper bound on the spending lost to other goods due to the price reduction, satisfies
\begin{eqnarray*}
&& 2(1-t^{1-\alpha}) p_i x_i \\
&\leq & 2 \left(1+\frac{\lambda(1+\delta)}{2(1-\lambda(1+\delta))}\right) (1-\alpha)(1-t) p_i  (1-2\delta)^{-\gamma/\beta} w_i \\
&\leq & \bar{\alpha} w_i |\Delta p_i|.
\end{eqnarray*}

\medskip

\noindent\textbf{Step 2:} Apply Lemma \ref{lem:change-PF} with the result of Step 1 to show that the potential function $\phi$ stays the same or decreases after a price update.

We assume $\Delta p_i > 0$. The proof is symmetric for $\Delta p_i < 0$. As the goods are pairwise complements, when $\Delta p_i > 0$, $S_{inc} = 0$ and $\Delta s_i = S_{dec}$.

\noindent\textbf{Case 2(a):}  $\mbox{sign}(x_i-\widetilde{w}_i)$ is not flipped and $x_i$ moves towards $\widetilde{w}_i$.

By Lemma \ref{lem:change-PF}, the change to $\phi$ is at most $-\widetilde{w}_i |\Delta p_i| + 2 |\Delta s_i| + c_1 \lambda p_i |\bar{x}_i - \widetilde{w}_i| + c_2 \delta w_i |\Delta p_i|$. Case 1(a) gives $2 |\Delta s_i| \leq \bar{\alpha} w_i |\Delta p_i|$. Noting that $\widetilde{w}_i/w_i\geq 1-\delta$, this change is at most $(\bar{\alpha} + c_1 + c_2 \delta - (1-\delta)) \lambda p_i |\bar{x}_i - \widetilde{w}_i|$. The second condition in this lemma implies this change is zero or negative.

\noindent\textbf{Case 2(b):} $\mbox{sign}(x_i-\widetilde{w}_i)$ is not flipped and $x_i$ moves away from $\widetilde{w}_i$, or $\mbox{sign}(x_i-\widetilde{w}_i)$ is flipped.

By Lemma \ref{lem:change-PF}, the change to $\phi$ is at most $-p_i |\bar{x}_i - \widetilde{w}_i| + \widetilde{w}_i |\Delta p_i| + c_1 \lambda p_i |\bar{x}_i - \widetilde{w}_i| + c_2 \delta w_i |\Delta p_i|$. Noting that $\widetilde{w}_i/w_i\leq 1+\delta$, this change is at most $\left((1+\delta+c_1+c_2\delta)\lambda - 1\right) p_i |\bar{x}_i - \widetilde{w}_i|$. The third condition in this lemma implies this change is zero or negative.
\end{proof}

\section{Bounds on the Warehouse Sizes}\label{app:bdd-warehouse-size}

\begin{rlemma}{lem:phase1-war-bdd}
In Phase 1, the total net change to $v_i$ is bounded by
$O(\frac{w_i}{\lambda} d(f_{\mbox{\tiny I}}) + \frac{w_i}{\lambda\beta} d(2) \log \frac{\beta}{\delta})$.
\end{rlemma}
\begin{proof}
In one day, $v_i$ shrinks by at most $(d(f) - 1) w_i$;
it can grow by at most $w_i$.
Since $f$ shrinks by a $1 - \Theta(\lambda)$ factor every $O(1)$ days
while $f \ge 2^{1/\beta}$,
during this part of Phase 1, $v_i$ can shrink by at most
$O(\sum_{i \ge 0} d(f_{\mbox{\tiny I}}[1 - \Theta(\lambda)]^i)w_i)
= O(\frac{w_i}{\lambda} d(f_{\mbox{\tiny I}}))$,
the equality following because $d(f)$ grows at least linearly.
In this part of Phase 1, $v(i)$ grows by at most
$w_i \log f_{\mbox{\tiny I}}
= O(w_i d(f_{\mbox{\tiny I}}))$.

The remainder of Phase 1 yields a further possible change of $d(2)$ to $v_i$ per day, for a
total of $O(\frac{w_i}{\lambda\beta} d(2) \log \frac{1}{\delta/\beta}) = O(\frac{w_i}{\lambda\beta} d(2) \log \frac{\beta}{\delta})$.
\end{proof}

\begin{rlemma}{lem:war-too-empt-impr}
Let $a_1, a_2, k > 0$.
Suppose that
$v_i \leq \vis - a_1 w_i $ and that $\kappa a_1 \ge 4 \lambda^2$.
Let $\tau$ be the time of a price update of $p_i$ to
$p_{i,1}$.
Suppose that henceforth $ p_{i} \leq e^{\bar{f}} p_{i,1} $
for some $\bar{f} \geq 0$.
If $k \ge \frac{2}{\kappa a_1}(\bar{f} + a_2)$, then
by time $\tau + (k+1)$ the warehouse
stock will have increased to more than $\vis - a_1 w_i $,
or by at least $a_2 w_i$, whichever is the lesser increase.
\end{rlemma}

\begin{lemma}\label{lem:war-too-full-impr}
Let $a_1, a_2, k > 0$.
Suppose that $v_i \geq \vis + a_1 w_i$.
Let $\tau$ be the time of a price update of $p_i$ to
$p_{i,1}$.
Suppose that henceforth $ p_{i} \geq e^{-\bar{f}} p_{i,1} $
for some $\bar{f} \geq 0$.
If $k \ge \frac{1}{\kappa a_1}(\bar{f} + a_2)$, then
by time $\tau + (k+1)$ the warehouse
stock will have decreased to less than $\vis + a_1 w_i $,
or by at least $a_2 w_i$, whichever is the lesser decrease.
\end{lemma}

\begin{proof} Suppose that $v_i \geq \vis + a_1 w_i $ throughout (or the result holds trivially).

Then each price change by a multiplicative $(1 + \mu \Delta t)$ is associated
with a target excess demand $\bar{z}_i = \bar{x}_i -w_i - \kappa (v_i -v_i^*)$,
where $\zbi = \mu w_i$.
Furthermore, the decrease to the warehouse stock since the previous price
update is exactly $(\xbi - w_i)\Delta t = [\mu w_i + \kappa (v_i -v_i^*)]\Delta t \ge
(\mu + \kappa a_1)w_i  \Delta t$.

Note that $1 + x \leq e^x$ for $|x| \leq 1$.
Thus $1 + \mu \Delta t \le e^{\mu \Delta t}$.

Suppose that over the next $k$ days there are $l-1$ price changes;
let the next $l$ price changes be by $1 + \mu_1 \Delta t_1, 1 + \mu_2 \Delta t_2,
\cdots, 1 + \mu_l \Delta t_l$.
Note that the total price change satisfies
$e^{- \bar{f}} \le \Pi_{1\le i \le l} (1 + \mu_i \Delta t_i)
\le e^{\sum_{1\le i \le l} \mu_i \Delta t_i}$.
Thus $\sum_{1\le i \le  l} \mu_i \Delta t_i \ge -\bar{f}$.

We conclude that when the $l$-th price change occurs,
the warehouse stock will have decreased by at least
$\sum_{1\le i \le l}(\mu_i  + \kappa a_1)w_i \Delta t_i \ge
(- \bar{f} + k\kappa a_1)w_i$.
If $k \ge \frac{1}{\kappa a_1}(\bar{f} + a_2)$, then
the warehouse stock decreases by at least $a_2 w_i$.
\end{proof}

\medskip

\begin{rtheorem}{lem:good-wrhs}
Suppose that the ratios $\chi_i/w_i$ are all equal.
Suppose that the prices are always $f_{\mbox{\tiny I}}$-bounded.
Also suppose that each price is updated at least once a day.
Suppose further that at the start of Phase 1 the warehouses are
all safe.
Finally, suppose that for all $i$:
\begin{enumerate}
\item
$\chi_i \ge \frac{512}{\beta} w_i$ and
$\chi_i \ge 8 v(f_{\mbox{\tiny I}}) w_i $.
\item
$\delta = \frac{\kappa \chi_i}{2 w_i}.$
\item
$\lambda^2 \le \frac {\kappa \chi_i} {32 w_i}$.
\end{enumerate}
Then the warehouse stocks never go outside their outer buffers
\emph{(}i.e.\ they never overflow
or run out of stock\emph{)};
furthermore,
after
$D(f_{\mbox{\tiny I}}) + \frac {32}{\beta} + \frac{2}{\kappa}$
days every warehouse will be safe thereafter.
\end{rtheorem}

\begin{proof}
We will consider warehouse $i$. We will say that $v_i$ lies in a particular zone
to specify how full or empty the warehouse is.

After $D(f_{\mbox{\tiny I}})$ days, Phase 2 has been reached.
By the first condition, in this period of time the warehouse stock can change by at most
$v(f_{\mbox{\tiny I}})w_i \le \chi_i/8$,
so $v_i$ can have moved out by at most one zone;
thus it lies in the middle buffer or a more central zone.

We show that henceforth the tendency is to improve, i.e.\ move toward the central zone,
but there can be fluctuations of up to one zone width. The result is that every
warehouse remains within its outer zone, and after a suitable time they will all
be in either their inner or central zone.

In Phase 2, the prices are $(1-2\delta)^{-1/\beta}$ bounded, we can conclude that
they are in the range $[1-2\delta/\beta, 1 + 4\delta/\beta]$ if $2\delta/\beta \le \frac 12$ and $\delta \le \frac 14$.
Further, this is contained in the range $[e^{-4\delta/\beta}, e^{4\delta/\beta}]$.
Hence $p_i$ can change by at most a factor of $e^{\pm 8\delta/\beta}$.

First we show that $v_i$ can move outward by at most one zone width. By Lemma \ref{lem:war-too-full-impr} (taking $a_1$ such that $a_1 w_i$ is the width of one zone, i.e.\ $a_1 = \frac 18 \chi_i/w_i$, $a_2=0$ and $\bar{f}=8\delta/\beta$),
after $8\delta/(\beta\kappa a_1)$ days the value of $v_i$ will have returned to value $v_i(t)$ or remained below this value.
During this period of time, the stock can increase by at most $8\delta w_i/(\beta\kappa a_1)$.
Note that by $\kappa \chi /2 = \delta w_i$, $a_1 = \frac 18 \chi_i/w_i$ and the first condition,
$8\delta w_i/(\beta\kappa a_1) = 32 w_i/\beta \leq \frac{1}{16} \chi_i$, which is half the width of a zone. This guarantees that the stock will never be overflow.

By Lemma \ref{lem:war-too-empt-impr} (taking $a_1 = \frac 18 \chi_i/w_i$, $a_2 = \frac 14 \chi_i/w_i$ and $\bar{f}=8\delta/\beta$), $v_i$ reaches the upper central zone after $(8\delta/\beta + a_2)/(\kappa a_1) = \frac{32}{\beta} + \frac{2}{\kappa}$ days. Applying the argument in the last paragraph anew shows that henceforth $v_i$ remains within the upper inner buffer.

We apply the same argument to the low zones using Lemma \ref{lem:war-too-empt-impr} (here we need to use the third condition).
The same results are achieved, but they take up to twice as long, and the possible
increase in stock is twice as large as the possible decrease in the previous
case, but still only one zone's worth.
\end{proof} 

\section{Unwinding the Conditions in the Complementary Case}\label{sect:unwinding1}

Lemma \ref{lem:PF-change-no-price-update}, Theorem \ref{lem:ongoing-price-change-1} and Theorem \ref{lem:good-wrhs} require several constraints on the parameters $\kappa, \delta, \lambda, c_1, c_2$.
We unwind these conditions to show how these parameters depend on the market parameters $\alpha,\gamma$ and $\beta$. We list the conditions below:

\begin{enumerate}
\item $4\kappa (1+c_2) \leq \lambda c_1 \leq 1/2$;
\item $(1-2\delta)^{-1/\beta} \leq (2-\delta)^{1/\gamma}$;
\item $\bar{\alpha} + c_1 + c_2 \delta \leq 1-\delta$ where
$$\bar{\alpha} = 2(1-\alpha)(1-2\delta)^{-\gamma / \beta} \left(1+\frac{\lambda(1+\delta)}{2(1-\lambda(1+\delta))}\right);$$
\item $(1+\delta+c_1+c_2\delta)\lambda \leq 1$;
\item $\chi_i \ge \frac{512}{\beta} w_i$ and $\chi_i \ge 8 v(f_{\mbox{\tiny I}}) w_i $;
\item $\delta = \frac{\kappa \chi_i}{2 w_i}$;
\item $\lambda^2 \le \frac {\kappa \chi_i} {32 w_i}$.
\end{enumerate}

Recall that without loss of generality we may assume $\chi_i/w_i$ are the same for all $i$.
Let $r = \frac{\chi_i}{w_i}$. When $r\geq \max\left\{\frac{512}{\beta},8v(f_I)\right\}$, Condition (5) is satisfied.

We first impose that
\begin{equation}\label{eq:unwind-11}
\delta \leq \min\left\{\frac{\beta}{2\gamma},\frac{1}{4}\right\}
, \lambda \leq \frac{3}{7}, c_1 = \delta, c_2 = 2.
\end{equation}
Condition (4) is then satisfied. Condition (1) becomes
\begin{equation}\label{eq:unwind-12}
\frac{24}{r}\leq \lambda.
\end{equation}

By Lemma \ref{lem:calculus-results}(e), $(1-2\delta)^{-\gamma/\beta} \leq 1 + \frac{4\gamma}{\beta}\delta$
as $2\gamma/\beta\leq 1/2$ and $\gamma/\beta\geq 1$.
Thus Condition (2) is satisfied when $1 + \frac{4\gamma}{\beta}\delta \leq 2-\delta$, which is equivalent to
\begin{equation}\label{eq:unwind-10}
\delta \leq \frac{\beta}{4\gamma + \beta}.
\end{equation}

Condition (3) is satisfied when $(8 + 4\gamma/\beta)\delta + 7\lambda/6 \leq \ln\frac{1}{2(1-\alpha)}$: this implies $(8 + 4\gamma/\beta)\delta + \frac{1+\delta}{2(1-\lambda(1+\delta))}\lambda \leq \ln\frac{1}{2(1-\alpha)}$ and hence $2(1-\alpha) \exp\left(4\delta\gamma/\beta + \frac{1+\delta}{2(1-\lambda(1+\delta))}\right) \leq \exp(-8\beta)$. This further implies
$$\bar{\alpha} = 2(1-\alpha)(1-2\delta)^{-\gamma/\beta}\left(1+\frac{1+\delta}{2(1-\lambda(1+\delta))}\right) \leq 1-4\delta.$$

When we further impose that
\begin{equation}\label{eq:unwind-13}
(8 + 4\gamma/\beta)\delta \leq \frac{1}{2}\ln\frac{1}{2(1-\alpha)},
\end{equation}
Condition (3) can be satisfied when
\begin{equation}\label{eq:unwind-14}
\lambda \leq \frac{3}{7}\ln\frac{1}{2(1-\alpha)}.
\end{equation}

Using the bounds on $\delta$ in \eqref{eq:unwind-11}, \eqref{eq:unwind-10} and \eqref{eq:unwind-13}, and substituting into Condition (6), yields
\begin{eqnarray*}
\kappa &\leq & \frac{2}{r}\cdot\min\left\{\frac{\beta}{2(\gamma+\beta)},\frac{1}{4}, \frac{\beta}{4\gamma + \beta}, \frac{1}{2(8 + 4\gamma/\beta)}\ln\frac{1}{2(1-\alpha)}\right\} \\
&=& \frac{2}{r}\cdot\min\left\{ \frac{\beta}{4\gamma + \beta}, \frac{1}{2(8 + 4\gamma/\beta)}\ln\frac{1}{2(1-\alpha)}\right\}.
\end{eqnarray*}

Using the bounds on $\lambda$ in \eqref{eq:unwind-11}, \eqref{eq:unwind-12} and \eqref{eq:unwind-14}, together with Condition (7), yields
$$\frac{24}{r} \leq \lambda \leq \min\left\{\frac{3}{7},\frac{3}{7}\ln\frac{1}{2(1-\alpha)},\sqrt{\frac{\kappa r}{32}}\right\}.$$
Note that $r = \chi_i / w_i$ needs to be sufficiently large to ensure that there is a choice of $\lambda$ which satisfies both the upper and lower bounds.

The market is defined by the parameters $\alpha,\gamma,\beta$. Then $\kappa,\lambda,r$ are chosen to satisfy the conditions.
The price update rule uses $\kappa,\lambda$, while the warehouse sizes are lower bounded by $r w_i$.
The parameters $c_1,c_2$ are needed only for the analysis.

\section{Markets with Mixtures of Substitutes and Complements}

\subsection{Phase 1 and One-Time Markets}

As with the case of markets of complementary goods, it suffices to analyze the one-time
markets in Phase 1.

\begin{lemma}\label{lem:f-bounded-demand-bound-2}
When the market is $f$-bounded,
\begin{enumerate}
\item if $p_i = r p_i^*/f$ where $1\leq r \leq f^2$, then $x_i \geq w_i f^\beta r^{-E}$;
\item if $p_i = f p_i^*/q$ where $1\leq q \leq f^2$, then $x_i \leq w_i f^{-\beta} q^E$.
\end{enumerate}
\end{lemma}

\begin{proof}
We prove the first part; the second part is symmetric.
Let $(p'_{-i},rp^*_i/f)$ be the $f$-bounded prices maximizing $x_i$
when $p_i = rp^*_i/f$.
First consider adjusting the prices from $p^*$ to $(p'_{-i},p^*_i/f)$
by smooth proportionate multiplicative changes (or equivalently, proportionate
linear changes to the terms $\log p_j$ for all $j$).
From  the definition of $\beta$ in Definition \ref{def:den-elas}, 
it is easy to show that the resulting demand for $x_i$ is at least $w_i f^{\beta}$.
Now increase $p_i$ by a factor of $r$.
As by assumption $E$ is the upper bound on the price elasticity,
the increase in the value of $p_i$ reduces $x_i$ by at most $r^{-E}$, yielding the bound 
$x_i \geq w_i f^\beta r^{-E}$.
\end{proof}

\begin{lemma}\label{lem:conv-one-time-market-2}
Suppose that $\beta > 0$. Further, suppose that the prices are updated independently using price update rule \eqref{eq:tat-rule},
 and that $0<\lambda\leq \frac{1}{2E-1}$. Let $p$ denote the current price vector and $p'$ denote the price vector after one day.
\begin{enumerate}
\item If $f(p)^\beta \geq 2$, then $f(p') \leq \left(1-\frac{\lambda}{2}\right) f(p)$.
\item If $f(p)^\beta \leq 2$, then $f(p') \leq f(p)^{1-\lambda\beta/(2\ln 2)}$.
\end{enumerate}
\end{lemma}

\begin{proof}
Suppose that $p_i = r \frac{p_i^*}{f(p)}$, where $1\leq r\leq f(p)^2$. By Lemma \ref{lem:f-bounded-demand-bound-2}, $x_i \geq w_i f(p)^\beta r^{-E}$ and hence $\frac{x_i-w_i}{w_i} \geq f(p)^\beta r^{-E} - 1$. When $p_i$ is updated using price update rule \eqref{eq:tat-rule}, the new price $p_i'$ satisfies
$$p_i' \geq r \frac{p_i^*}{f(p)} \left[1 + \lambda\cdot\min\left\{1,f(p)^\beta r^{-E} - 1\right\}\right].$$
Let $h_3(r) := r \left[1 + \lambda\cdot\min\left\{1,(f(p)^\beta r^{-E} - 1)\right\}\right]$. When $f(p)^\beta r^{-E}\geq 2$, $\frac{d h_3(r)}{dr} = 1+\lambda > 0$; When $f(p)^\beta r^{-E}\leq 2$,
$$\frac{d}{dr} h_3(r) = 1 - \lambda\left(1+(E-1) f(p)^\beta r^{-E}\right) \geq 1-\lambda (2E-1) \geq 0.$$
Thus
$$p_i' \geq \frac{p_i^*}{f(p)} \left[1 + \lambda\cdot\min\left\{1,f(p)^\beta - 1\right\}\right].$$

Similarly, suppose $p_j$ satisfies $p_j = \frac{1}{q} f(p) p_j^*$, where $1\leq q\leq f(p)^2$. By Lemma \ref{lem:f-bounded-demand-bound-2}, $x_j \leq w_j f(p)^{-\beta} q^E$ and hence $\frac{x_j-w_j}{w_j}\leq f(p)^{-\beta} q^E-1$. When $p_j$ is updated using price update rule \eqref{eq:tat-rule}, the new price $p_j'$ satisfies
$$p_j' \leq \frac{1}{q} f(p) p_j^* \left( 1 + \lambda\cdot\min\left\{1,f(p)^{-\beta} q^E-1\right\} \right).$$
Let $h_4(q) := \frac{1}{q}\left[1 + \lambda\cdot\min\left\{1,f(p)^{-\beta} q^E-1\right\}\right]$. When\\
$f(p)^{-\beta} q^E \geq 2$, $\frac{d}{dq} h_4(q) = -\frac{1}{q^2}(1+\lambda) < 0$. When $f(p)^{-\beta} q^E \leq 2$,
\begin{eqnarray*}
\frac{d}{dq} h_4(q) &=& \frac{1}{q^2}\left[-1 + \lambda\left(1+(E-1) f(p)^{-\beta} q^E\right)\right] \\
&\leq & \frac{1}{q^2} \left(-1 + \lambda(2E-1)\right) \leq 0.
\end{eqnarray*}
Thus
\begin{eqnarray*}
p_j' &\leq & f(p) p_j^* \left[ 1 + \lambda\cdot\min\left\{ 1, f(p)^{-\beta} -1\right\} \right] \\
&=& f(p) p_j^* \left[ 1 + \lambda (f(p)^{-\beta} -1) \right].
\end{eqnarray*}

The remainder of the proof is exactly same as the the final part of the proof of Lemma \ref{lem:conv-one-time-market-1}.
\end{proof}

\subsection{Ongoing Markets}

\begin{lemma}\label{lem:bound-M2}
Suppose Assumption 1 holds, $\lambda E\leq 1$ and $\lambda\leq \frac 12$, then $|\Delta S_s| \leq (\alpha' + 2(E-1)) x_i |\Delta p_i|$.
\end{lemma}

\begin{proof}There are two cases.

\noindent\textbf{Case 1.} The price of $G_i$ is reduced from $p_i$ to $p_i' = t p_i$, where $t<1$.

Then $x_i' \leq t^{-E} x_i$ and hence $\Delta s_i \leq (t^{1-E}-1) p_i x_i$. Then
\begin{eqnarray*}
|\Delta S_s| =  |\Delta S_c| + \Delta s_i &\leq & \alpha' x_i |\Delta p_i| + (t^{1-E}-1) p_i x_i \\
&\leq & \alpha' x_i |\Delta p_i| + 2(E-1)(1-t) p_i x_i.
\end{eqnarray*}
The last inequality holds by applying Lemma \ref{lem:calculus-results}(c) with $\max\left\{\frac{E\lambda}{2},\lambda\right\}\leq 1/2$ and $t\geq 1-\lambda$. Noting that $(1-t) p_i = |\Delta p_i|$, completes the proof.

\noindent\textbf{Case 2.} The price of $G_i$ is raised from $p_i$ to $p_i' = t p_i$, where $t>1$.

Then $x_i' \geq t^{-E} x_i$ and hence $\Delta s_i \geq (t^{1-E}-1) p_i x_i$. Then
\begin{eqnarray*}
|\Delta S_s| = |\Delta S_c| - \Delta s_i &\leq & \alpha' x_i |\Delta p_i| + (1-t^{1-E}) p_i x_i\\
&\leq & \alpha' x_i |\Delta p_i| + (E-1)(t-1) p_i x_i.
\end{eqnarray*}
The last inequality holds by applying Lemma \ref{lem:calculus-results}(d). Noting that $(t-1) p_i = |\Delta p_i|$, completes the proof.
\end{proof}

The following lemma proves convergence in the market with mixtures of substitutes and complements while incorporating warehouses.

\smallskip

\begin{rlemma}{lem:ongoing-price-change-2}
Suppose Assumption 1 holds and $\beta$, as defined in Definition \ref{def:den-elas}, satisfies $\beta>0$. Suppose that the following conditions hold:
\begin{enumerate}
\item $f \leq (1-2\delta)^{-1/\beta} \leq (2-\delta)^{1/(2E-\beta)}$ since the last price update to $p_i$;
\item $2\alpha '(1-2\delta)^{-(2E-\beta)/\beta} + c_1 + c_2 \delta \leq 1-\delta$;
\item $\left(2\alpha ''(1-2\delta)^{-(2E-\beta)/\beta} + 1 + \delta + c_1 + c_2 \delta\right)\lambda \leq 1$,
\end{enumerate}
where $\alpha '' := \alpha' + 2(E-1)$.
Then, when a price $p_i$ is updated using rule \eqref{eq:tat-rule-warehouse}, the value of $\phi$ stays the same or decreases.
\end{rlemma}

\begin{proof}
The first condition is used with Lemma \ref{lem:f-bounded-demand-bound-2} to ensure that 
$x_i, \bar{x}_i  \leq (1-2\delta)^{-(2E-\beta)/\beta} w_i \leq (2-\delta) w_i$, and hence that $\frac{\bar{z}_i}{w_i}\leq 1$. By price update rule \eqref{eq:tat-rule-warehouse}, $|\Delta p_i|=\lambda p_i |\bar{x}_i - \widetilde{w}_i| / w_i$.

We assume $\Delta p_i > 0$. The proof is symmetric for $\Delta p_i < 0$. Recall that when $\Delta p_i > 0$, $|\Delta S_s| = |\Delta S_c| - \Delta s_i$.

\medskip

\noindent\textbf{Case 1:} $\mbox{sign}(x_i-\widetilde{w}_i)$ is not flipped and $x_i$ moves towards $\widetilde{w}_i$.

By Lemma \ref{lem:change-PF}, the change to $\phi$ is at most $-\widetilde{w}_i |\Delta p_i| + \Delta s_i + |\Delta S_c| + |\Delta S_s| + c_1 \lambda p_i |\bar{x}_i - \widetilde{w}_i| + c_2 \delta w_i |\Delta p_i|$, which is equal to $-\widetilde{w}_i |\Delta p_i| + 2 |\Delta S_c| + c_1 \lambda p_i |\bar{x}_i - \widetilde{w}_i| + c_2 \delta w_i |\Delta p_i|$. Noting $\widetilde{w}_i/w_i \geq 1-\delta$ and $x_i \leq (1-2\delta)^{-(2E-\beta)/\beta} w_i$, 
and applying Assumption~\ref{ass:sp-trans}, implies that 
this change is at most $\left(2 \alpha '(1-2\delta)^{-(2E-\beta)/\beta}+c_1+c_2\delta-(1-\delta)\right) w_i |\Delta p_i|$.

\noindent\textbf{Case 2:} $\mbox{sign}(x_i-\widetilde{w}_i)$ is not flipped and $x_i$ moves away from $\widetilde{w}_i$, or $\mbox{sign}(x_i-\widetilde{w}_i)$ is flipped.

By Lemma \ref{lem:change-PF}, the change to $\phi$ is at most $-p_i |\bar{x}_i - \widetilde{w}_i| + \widetilde{w}_i |\Delta p_i| - \Delta s_i + |\Delta S_c| + |\Delta S_s| + c_1 \lambda p_i |\bar{x}_i - \widetilde{w}_i| + c_2 \delta w_i |\Delta p_i|$, which is equal to $-p_i |\bar{x}_i - \widetilde{w}_i| + \widetilde{w}_i |\Delta p_i| + 2 |\Delta S_s| + c_1 \lambda p_i |\bar{x}_i - \widetilde{w}_i| + c_2 \delta w_i |\Delta p_i|$. Noting that $\widetilde{w}_i/w_i \leq 1+\delta$ and $x_i \leq (1-2\delta)^{-(2E-\beta)/\beta} w_i$, 
and applying Assumption~\ref{ass:sp-trans}, implies that this change is at most
$$\left((2\alpha ''(1-2\delta)^{-(2E-\beta)/\beta}+1+\delta+c_1+c_2\delta)\lambda - 1\right) p_i |\bar{x}_i - \widetilde{w}_i|.$$
\end{proof}

\section{Unwinding the Conditions in the Mixture Case}\label{sect:unwinding2}

Lemma \ref{lem:PF-change-no-price-update}, Lemma \ref{lem:ongoing-price-change-2} and Theorem \ref{lem:good-wrhs} require several constraints on the parameters $\kappa, \delta, \lambda, c_1, c_2$.
We unwind these conditions to show how these parameters depend on the market parameters $\beta$ and $\alpha '$. We list the conditions below:

\begin{enumerate}
\item $4\kappa (1+c_2) \leq \lambda c_1 \leq 1/2$;
\item $(1-2\delta)^{-1/\beta} \leq (2-\delta)^{1/(2E-\beta)}$;
\item $2\alpha '(1-2\delta)^{-(2E-\beta)/\beta} + c_1 + c_2 \delta \leq 1-\delta$;
\item $\left(2\alpha ''(1-2\delta)^{-(2E-\beta)/\beta} + 1 + \delta + c_1 + c_2 \delta\right)\lambda \leq 1$;
\item $\chi_i \ge \frac{512}{\beta} w_i$ and $\chi_i \ge 8 v(f_{\mbox{\tiny I}}) w_i $;
\item $\delta = \frac{\kappa \chi_i}{2 w_i}$;
\item $\lambda^2 \le \frac {\kappa \chi_i} {32 w_i}$.
\end{enumerate}

We first impose the conditions
\begin{equation}\label{eq:unwind-21}
\delta \leq \min\left\{\frac{\beta}{4(2E - \beta)},\frac{1}{4} \right\}, \lambda\leq 1,c_1 = \delta, c_2 = 2.
\end{equation}
As in Section \ref{sect:unwinding1}, let $r = \frac{\chi_i}{w_i}$. When  $r\geq \max\left\{\frac{512}{\beta},8v(f_I)\right\}$ and
\begin{equation}\label{eq:unwind-20}
\lambda\geq\frac{24}{r},
\end{equation}
Conditions (5) and (1) are satisfied.

By Lemma \ref{lem:calculus-results}(e), $(1-2\delta)^{-(2E-\beta)/\beta} \leq 1 + \frac{4(2E-\beta)}{\beta}\delta$ as
$\frac{2(2E - \beta)\delta}{\beta} \le 1/2$ and $\frac{2E-\beta}{\beta}\geq 1$.
Thus Condition (2) and (3) are satisfied when $1 + \frac{4(2E-\beta)}{\beta}\delta \leq 2-\delta$ and $2\alpha '\left(1 + \frac{4(2E-\beta)}{\beta}\delta\right) + 4\delta \leq 1$ respectively, which are equivalent to
\begin{equation}\label{eq:unwind-22}
\delta \leq \frac{\beta}{\beta+4(2E-\beta)}, \delta \leq \frac{(1-2\alpha ')\beta}{8\alpha '(2E-\beta) + 4\beta}.
\end{equation}

Condition (4) is satisfied when
$$\left(2\alpha '' \left(1 + \frac{4(2E-\beta)}{\beta}\delta\right) + 1 + 4\delta\right)\lambda\leq 1.$$
The bounds on $\delta$ in \eqref{eq:unwind-21} gives $\frac{4(2E-\beta)}{\beta}\delta\leq 1$ and $4\delta\leq 1$, hence Condition (4) is satisfied when
\begin{equation}\label{eq:unwind-23}
\lambda \leq \frac{1}{4\alpha ''+2} = \frac{1}{8E + 4\alpha ' - 6}.
\end{equation}

Using the bounds on $\delta$ in \eqref{eq:unwind-21} and \eqref{eq:unwind-22},
and substituting into Condition (6), yields
\begin{eqnarray*}
\kappa  \leq \frac{2}{r}\cdot\min\left\{ \frac{\beta}{\beta+4(2E-\beta)},\frac{(1-2\alpha ')\beta}{8\alpha '(2E-\beta) + 4\beta}\right\}.
\end{eqnarray*}

Using the bounds on $\lambda$ in \eqref{eq:unwind-21}, \eqref{eq:unwind-20} and \eqref{eq:unwind-23}, together with Condition (7), yields
$$\frac{24}{r}\leq \lambda \leq \min\left\{\frac{1}{8E + 4\alpha ' - 6},\sqrt{\frac{\kappa r}{32}}\right\}.$$

The market is defined by the parameters $E$, $\beta$ and $\alpha '$. Then $\kappa,\lambda,r$ are chosen to satisfy the conditions.
The price update rule uses $\kappa,\lambda$, while the warehouse sizes are lower bounded by $r w_i$.
The parameters $c_1,c_2$ are needed only for the analysis.

\subsection{Example Scenario: \boldmath{$N$}-Level Nested CES Type Utilities}

We focus on one particular good $i$.
Let $A_1,A_2,\cdots,A_N$ be the square nodes along the path from good $i$ to the root of the utility tree,
and let $\rho_1,\rho_2,\cdots,\rho_N$ be the associated $\rho$ values.
Let $\sigma_k = \frac{1}{1-\rho_k}$ for $1\leq k\leq N$.
Let $S_k$ denote the set of goods which are in the subtree rooted at $A_k$. 
Let $h_k$ denote the total spending on all goods in $S_k$ and let
$ANC(j)$ denote the least common ancestor of goods $i$ and $j$.

Keller derived the following formulae:
\begin{eqnarray*}
\frac{\partial x_i/\partial p_j}{x_i / p_j} &=& \frac{s_j}{h_N} (\sigma_N - 1) + \sum_{q=ANC(j)}^{N-1} \frac{s_j}{h_q} (\sigma_q - \sigma_{q+1})\\
\frac{\partial x_i/\partial p_i}{x_i / p_i} &=& -\sigma_1 + \frac{s_i}{h_N}(\sigma_N - 1) + \sum_{q=1}^{N-1} \frac{s_i}{h_q} (\sigma_q - \sigma_{q+1}).
\end{eqnarray*}
We now compute the Adverse Market Elasticity of good $i$.  When the price of good $i$ is reduced by a factor of $(1-\delta)$ (think of $\delta$ as being very small), raise the prices of all the complements of good $i$ by a factor of $1/(1-\delta)$ and reduce the prices of all the substitutes of good $i$ by a factor of $(1-\delta)$. By the above formulae, $x_i' \geq x_i t^{\beta_i}$, where
\begin{eqnarray*}
\beta_i &=& -\frac{\partial x_i/\partial p_i}{x_i / p_i} - \sum_{j\neq i} \left| \frac{\partial x_i/\partial p_j}{x_i / p_j} \right|\\
&=& \sigma_1 - \frac{s_i}{h_N}(\sigma_N - 1) - \sum_{q=1}^{N-1} \frac{s_i}{h_q} (\sigma_q - \sigma_{q+1}) \\
&& \hspace{0.3in} - \sum_j \left| \frac{s_j}{h_N} (\sigma_N - 1) + \sum_{q=ANC(j)}^{N-1} \frac{s_j}{h_q} (\sigma_q - \sigma_{q+1}) \right|\\
&\geq & \sigma_1 - \frac{s_i}{h_N}|\sigma_N - 1| - \sum_{q=1}^{N-1} \frac{s_i}{h_q} |\sigma_q - \sigma_{q+1}| \\
&& \hspace{0.3in} - \sum_j \left( \frac{s_j}{h_N} |\sigma_N - 1| + \sum_{q=ANC(j)}^{N-1} \frac{s_j}{h_q} |\sigma_q - \sigma_{q+1}| \right)\\
&=& \sigma_1 - |\sigma_N - 1|\left(\sum_{j\in S_N} \frac{s_j}{h_N}\right)- \sum_{q=1}^{N-1} \left(|\sigma_q - \sigma_{q+1}| \sum_{j\in S_q} \frac{s_j}{h_q}\right)\\
&=& \sigma_1 - |\sigma_N-1| - \sum_{q=1}^{N-1} |\sigma_q - \sigma_{q+1}|.
\end{eqnarray*}
Note that we do not require any two goods to always be substitutes or always complements.

The lower bound on $\beta_i$ is tight when $\frac{h_N}{h_{N_1}}, \frac{h_{N-1}}{h_{N-2}},\cdots,\frac{h_2}{h_1},\frac{h_1}{s_i}$ are all very large. Set $\beta$, as defined in Definition \ref{def:den-elas}, to $\min_i \beta_i$.

Also note that
\begin{eqnarray*}
\frac{\partial x_i/\partial p_i}{x_i / p_i} &=& \frac{s_i}{h_N}\left[-1-\sum_{q=1}^N \sigma_q\left(\frac{h_N}{h_{q-1}} - \frac{h_N}{h_q}\right)\right]\\
&\geq & \max\left\{1,\max_{1\leq k \leq N} \sigma_k\right\}\left(-\frac{s_i}{h_N}-\sum_{q=1}^N \left(\frac{s_i}{h_{q-1}} - \frac{s_i}{h_q}\right)\right)\\
&=& -\max\left\{1,\max_{1\leq k \leq N} \sigma_k\right\}.
\end{eqnarray*}
Let $E_i = \max\left\{1,\max_{1\leq k \leq N} \sigma_k\right\}$ and set $E = \max_i E_i$.

Keller also derived that
$$\frac{\partial s_j}{\partial p_i} = x_i\left(\frac{s_j}{h_N} (\sigma_N - 1) + \sum_{q=ANC(j)}^{N-1} \frac{s_j}{h_q} (\sigma_q - \sigma_{q+1})\right).$$
This yields
\begin{eqnarray*}
&& \sum_{j\neq i} |\Delta s_j|\\
&\leq & (1-\lambda)^{-E} x_i |\Delta p_i| \sum_{j\neq i}\left(\frac{s_j}{h_N} |\sigma_N - 1| + \sum_{q=ANC(j)}^{N-1} \frac{s_j}{h_q} |\sigma_q - \sigma_{q+1}|\right) \\
&\leq & x_i |\Delta p_i| (1-\lambda)^{-E} \left(|\sigma_N-1| + \sum_{q=1}^{N-1} |\sigma_q - \sigma_{q+1}|\right).
\end{eqnarray*}
Let $\alpha_i ' = (1-\lambda)^{-E} \left(|\sigma_N-1| + \sum_{q=1}^{N-1} |\sigma_q - \sigma_{q+1}|\right)$. Assumption \ref{ass:sp-trans} is satisfied with $\alpha' = \max_i \alpha_i '$.
Hence the bounds from Theorems~\ref{th:multi-overall} and~\ref{lem:good-wrhs} apply.

\end{document}